\documentclass[a4paper,USenglish,cleveref,autoref,numberwithinsect]{lipics-v2019}

\title{Typed $\lambda$-calculi and superclasses of regular transductions}

\titlerunning{}

\author{Lê Thành D\~{u}ng Nguy\~{ê}n}{LIPN,
  UMR 7030 CNRS, Université Paris 13, France \and \url{https://nguyentito.eu/}}{nltd@nguyentito.eu}{https://orcid.org/0000-0002-6900-5577}{}

\authorrunning{L.~T.~D.~Nguy\~{ê}n}

\Copyright{John Q. Public and Joan R. Public}



\category{}

\supplement{}

\acknowledgements{Thanks to Pierre Pradic for bringing to my attention the
  question (asked by Mikołaj Bojańczyk) of relating transductions and linear
  logic and for the proof of \cref{thm:polyreg-hdt0l}. This work also benefited
  from discussions on automata theory with Célia Borlido, Marie Fortin and
  Jérémy Ledent, and on the simply typed $\lambda$-calculus with Damiano
  Mazza.}

\nolinenumbers 

\hideLIPIcs  

\EventEditors{John Q. Open and Joan R. Access}
\EventNoEds{2}
\EventLongTitle{42nd Conference on Very Important Topics (CVIT 2016)}
\EventShortTitle{CVIT 2016}
\EventAcronym{CVIT}
\EventYear{2016}
\EventDate{December 24--27, 2016}
\EventLocation{Little Whinging, United Kingdom}
\EventLogo{}
\SeriesVolume{42}
\ArticleNo{23}

\usepackage{cmll}
\usepackage{stmaryrd}
\usepackage{csquotes}
\usepackage{bussproofs}

\theoremstyle{definition}
\newtheorem{notation}[theorem]{Notation}
\newtheorem{question}[theorem]{Open question}

\newcommand{\Str}{\mathtt{Str}}
\newcommand{\BT}{\mathtt{BT}}
\newcommand{\dBT}{{\partial\mathtt{BT}}}
\newcommand{\Bool}{\mathtt{Bool}}
\newcommand{\Nat}{\mathtt{Nat}}
\newcommand{\Fin}{\mathtt{Fin}}

\newcommand{\STlam}{{\mathrm{ST}\lambda}}
\newcommand{\EAlam}{{\mathrm{EA}\lambda}}
\newcommand{\muEAlam}{{\mu\mathrm{EA}\lambda}}

\newcommand{\naturalN}{\mathbb{N}}
\newcommand{\Lang}{\mathcal{L}}

\newcommand{\Compile}{\mathcal{C}}
\newcommand{\CbS}{\mathrm{CbS}}
\newcommand{\BinTree}{\mathrm{BinTree}}
\newcommand{\dBinTree}{\partial\mathrm{BinTree}}
\newcommand{\ExprBT}{\mathrm{ExprBT}}
\newcommand{\ExprdBT}{\mathrm{Expr}\partial\mathrm{BT}}
\newcommand{\ttl}{\triangleleft}
\newcommand{\ttr}{\triangleright}

\begin{document}
\maketitle

\begin{abstract}
  We propose to use Church encodings in typed $\lambda$-calculi as the basis for
  an automata-theoretic counterpart of implicit computational complexity, in the
  same way that monadic second-order logic provides a counterpart to descriptive
  complexity. Specifically, we look at transductions i.e.\ string-to-string (or
  tree-to-tree) functions -- in particular those with superlinear growth, such
  as polyregular functions, HDT0L transductions and Sénizergues's
  \enquote{$k$-computable mappings}.

  Our first results towards this aim consist showing the inclusion of some
  transduction classes in some classes defined by $\lambda$-calculi. In
  particular, this sheds light on a basic open question on the expressivity of
  the simply typed $\lambda$-calculus. We also encode regular functions (and, by
  changing the type of programs considered, we get a larger subclass of
  polyregular functions) in the elementary affine $\lambda$-calculus, a variant
  of linear logic originally designed for implicit computational complexity.
\end{abstract}

\section{Introduction}

The main goal of this paper is to provide some evidence for connections between:
\begin{itemize}
\item automata theory, in particular transducers (loosely defined as devices
  which compute string-to-string (or tree-to-tree) functions and are
  \enquote{finite-state} in some way);
\item programming language theory, in particular the expressive power of some
  \emph{typed $\lambda$-calculi}, i.e.\ some (minimalistic) statically typed
  functional programming languages.
\end{itemize}
Our first concrete result is:
\begin{theorem}
  \label{thm:stlam-intro}
  The functions from strings to strings that can be expressed (in a certain way)
  in the \emph{simply-typed $\lambda$-calculus} ($\STlam$) -- we shall call
  these the \emph{$\lambda$-definable string functions}
  (\cref{def:lambda-definable}) -- enjoy the following properties:
  \begin{itemize}
  \item they are closed under composition;
  \item they are \emph{regularity-preserving}: the inverse image of a \emph{regular
      language} is regular;
  \item they contain all the transductions defined by \emph{HDT0L systems
      (see~\cite{Senizergues,FiliotReynier})}, a variant of the \emph{L-systems}
    originally introduced by Lindenmayer for mathematical
    biology~\cite{Lindenmayer}.
  \end{itemize}
\end{theorem}
We believe that this is conceptually interesting for the study of both
$\lambda$-calculi and automata:
\begin{itemize}
\item It is directly relevant to a basic and natural open problem about the
  functions $\naturalN \to \naturalN$ definable (in some way) in $\STlam$. This
  problem is simple enough to be presented without assuming any background in
  programming language theory; we shall do this in
  \S\ref{sec:intro-lambda-definable}.
\item Another corollary is that the simply typed $\lambda$-calculus subsumes all
  the natural classes of regularity-preserving functions that we know of. We
  indeed prove in this paper that the \emph{closure by composition of HDT0L
    transductions} (a class that we shall abbreviate as
  \enquote{HDT0L+composition}) contains the \emph{polyregular functions}
  recently introduced by Bojańczyk~\cite{polyregular}; therefore, it includes
  \emph{a fortiori} the well-known classes of \emph{regular}, \emph{rational}
  and \emph{sequential} string functions (see e.g.~\cite{siglog,MuschollPuppis},
  or the introduction to~\cite{polyregular}).
\end{itemize}

The above-mentioned classes can be defined using transducers, and they admit
alternative characterizations which attest to their robustness. For instance,
regular functions can be also characterized by Monadic Second-Order
Logic~\cite{EngelfrietHoogeboom}.

\subparagraph{The general pattern: encoding transductions}

More generally, several results in this paper consist in considering, on one
hand, some class $\mathcal{A}$ of automata with output, and on the other hand,
some typed $\lambda$-calculus $\mathcal{P}$ with a type $T$. The programs of
type $T$ in $\mathcal{P}$ must be able to take (encodings of) strings as inputs
and output (encodings of) strings. Then \enquote{compiling} the automata in
$\mathcal{A}$ into programs in $\mathcal{P}$, we get:
\[
  \begin{Bmatrix}
    \text{functions computed}\\
    \text{by transducers in $\mathcal{A}$}  
  \end{Bmatrix}
  \subseteq
  \begin{Bmatrix}
    \text{functions computed}\\
    \text{by programs in $\mathcal{P}$ of type $T$}
  \end{Bmatrix}
\]
Of course, an equality sign here would be more satisfying. But we do not know
whether all the $\lambda$-definable string functions (in $\STlam$) are in
HDT0L+composition. In contrast, for our next result, even though we only claim
and prove an inclusion in this paper, we are actually fairly confident that the
converse holds. But it appears to be significantly more difficult than the
direction treated here: our tentative proof\footnote{A joint work with Paolo
  Pistone, Thomas Seiller and Lorenzo Tortora de Falco. We intend to present
  this work in a future paper -- hence the numbering in the title.} for this
converse -- which has not been thoroughly checked -- requires the development of
new tools in denotational semantics.

\subparagraph{Linear logic vs streaming string transducers}

This next result involves the \emph{elementary affine $\lambda$-calculus}
($\EAlam$) introduced by Baillot, De~Benedetti and Ronchi Della
Rocca~\cite{Benedetti}.
\begin{theorem}
  The programs of a certain type in $\EAlam$ compute all \emph{regular
    functions}, and compute only \emph{linear time} and
  \emph{regularity-preserving} functions.
\end{theorem}
See \cref{thm:ealam} for a precise statement. $\EAlam$ is mainly inspired by
Girard's \emph{linear logic}~\cite{girardLL}, a \enquote{resource-sensitive}
constructive logic that has already been used to characterize complexity classes
(see \S\ref{sec:intro-other}). In programming languages, \emph{linearity} refers
to the prohibition of \emph{duplication}: a function is linear if it uses its
argument at most\footnote{Strictly speaking, such a function is \emph{affine}; a
  linear function uses its argument \emph{exactly} once. But we follow here a
  widespread abuse of language.} once. Linearity appears in automata theory
under the name\footnote{The term \enquote{linearity} itself has also been used,
  e.g.\ in~\cite{siglog}: \enquote{updates should make a linear use of
    registers}.} \enquote{copyless assignment}. This refers to a technical
condition in the definition of \emph{streaming string transducers} (SSTs), a
machine model introduced by Alur and Černý~\cite{SST}. Hence the relevance of
the elementary affine $\lambda$-calculus to regular functions:
\begin{itemize}
\item the functions computed by SSTs are exactly the regular functions;
\item without the linearity condition, the class obtained is instead the HDT0L
  transductions, as proved recently by Filiot and Reynier~\cite{FiliotReynier}.
\end{itemize}
Thus, our work gives a precise technical contents to this analogy between
linearity in $\lambda$-calculi and copyless assignments in automata theory: what
makes \enquote{copyful SSTs} impossible to encode in $\EAlam$ is their
non-linearity.

\subparagraph{String functions of superlinear growth}

In the above theorem, instead of the converse inclusion, we have merely stated
an upper bound on the definable string functions in $\EAlam$, in terms of time
complexity. This already means that we capture a much smaller class of functions
than in our previous result on $\STlam$. Indeed, since HDT0L transductions can
grow exponentially, when one composes them, the rates of growth can become
towers of exponentials. However, the other classes that we mentioned contain
only tractable functions: not only do they grow polynomially, they are also
computable in polynomial time. In particular the regular functions are computed
in linear time -- and we match this bound.

Another instance of our pattern takes place in the same language $\EAlam$. By
changing the type of programs considered, we manage to code a larger subclass of
polyregular functions -- it contains functions whose output length may grow
polynomially with arbitrary exponent. As an added benefit, this partially
answers a natural question concerning $\EAlam$ (we discuss this further in
\S\ref{sec:intro-other}).

With this last result, together with \cref{thm:stlam-intro}, we hope to
contribute to the recent surge of interest in superlinear transductions,
exemplified by the introduction of polyregular functions~\cite{polyregular} --
whose slogan is \enquote{the polynomial growth finite state transducers} -- and
the study of non-linear streaming string transducers. Concerning the latter,
HDT0L systems were mostly used to describe \emph{languages} previously; in fact,
before Filiot and Reynier's work~\cite{FiliotReynier}, their semantics as
transductions seems to have been considered only once: in an invited paper
without proofs~\cite{Senizergues}, Sénizergues claims to characterize the
HDT0L+composition class using iterated pushdown automata.

\subparagraph{Tree transductions}

Finally, we shall also see that in all the above programming languages, there is
a type of functions from binary trees to binary trees, and all \emph{regular
  tree functions} can be encoded as programs of this type. This relies on their
characterization by \emph{bottom-up ranked tree transducers}~\cite{BRTT}
generalizing SSTs with a relaxed and more subtle linearity condition -- closely
related, as we shall see, to the \emph{additive conjunction} of linear logic
(while the linearity of SSTs is purely \emph{multiplicative}). We do not
investigate superlinear tree transducers here.

\subparagraph{Plan of the paper}

In the remainder of this introduction, we first briefly present the simply typed
$\lambda$-calculus, and state our motivating problem on the functions $\naturalN
\to \naturalN$ that it can express (\S\ref{sec:intro-lambda-definable}). The
other introductory subsection (\S\ref{sec:intro-other}) situates our work in the
conceptual landscape, and surveys related work and inspirations.
\Cref{sec:transducers} introduces the automata-theoretic classes of functions
studied here, and proves some inclusions between them. \Cref{sec:stlam} and
\ref{sec:ealam} are dedicated respectively to $\STlam$ and $\EAlam$.

\subparagraph{Intended audience}

We have attempted to make the parts involving the simply typed
$\lambda$-calculus accessible to a broad audience, since the arguments involved
are rather elementary. However, the exposition of the results on $\EAlam$
assumes some familiarity with linear logic.

\subsection{Motivation: $\lambda$-definable numeric functions}
\label{sec:intro-lambda-definable}

\subsubsection{Introduction to the $\lambda$-calculus and to Church encodings}

The untyped $\lambda$-calculus is a naive syntactic theory of functions. Its
terms are generated by the grammar\footnote{The terms must actually be
  considered up to renaming of bound variables, just as usual mathematical
  practice dictates that $x$ is bound in $t$ in the expression $x \mapsto t$.
  The details of this renaming equivalence, called
  \enquote{$\alpha$-conversion}, are uninteresting and can be found in any
  textbook on the $\lambda$-calculus. Similarly, in the substitution $t\{x:=u\}$
  introduced later, only the free occurrences of $x$, i.e.\ not appearing under
  a $\lambda x.$, must be substituted.} $t,u ::= x \mid t\; u \mid \lambda x.\;
t$ (where $x$ is taken in a countable set of \enquote{variables}), which mirrors
the basic operations of function application ($t\; u \approx t(u)$) and function
formation ($\lambda x.\; t \approx x \mapsto t$). The equational theory on these
$\lambda$-terms is the congruence generated by
\[ (\lambda x.\; t)\; u =_\beta t\{x := u\} \qquad \text{where $t\{x:=u\}$ is
    the substitution of $x$ by $u$ in $t$}\]
which corresponds to the usual way of computing a function, e.g.\ $(x \mapsto
x^2 + 1)(42) = 42^2 + 1$.

This example cannot be directly expressed in the $\lambda$-calculus since it
does not have primitive integers among its terms. Instead, we use \emph{Church
  encodings} to represent natural numbers: morally, $n \in \naturalN$ is encoded
as the $n$-fold iteration functional $\overline{n}: f \mapsto f^{\,n} = f \circ
\ldots \circ f$. For instance, $\overline{2} = \lambda f.\; (\lambda x.\; f\;
(f\; x))$. Using this encoding, the untyped $\lambda$-calculus can represent any
computable function $f : \naturalN \to \naturalN$: for some term $t$, $t\;
\overline{n} =_\beta \overline{f(n)}$ for all $n \in \naturalN$.

To avoid the pitfalls of Turing-completeness (e.g.\ to obtain only total
functions), one technique is to add a \emph{type system}: a way of annotating
terms with \emph{types} specifying some of their behavior. In the simply typed
$\lambda$-calculus ($\STlam$), we use the \emph{simple types} defined as $A, B
:= o \mid A \to B$, where $o$ is the single \emph{base type}. We write $t : A$
when the term $t$ can be given the type $A$. The meaning of $t : A \to B$ is
morally that $t$ is a function taking inputs of type $A$ and returning outputs
of type $B$.

The rules of $\STlam$ allow us for example to show that assuming $f : o \to o$
and $x : o$, we have $f\; (f\; x) : o$; from this, one can then deduce that
$\overline{2} = \lambda f.\; (\lambda x.\; f\; (f\; x)) : (o \to o) \to (o \to
o)$ (without assumption). In general, one can show that the terms of type $\Nat
= (o \to o) \to (o \to o)$, quotiented by $=_\beta$, are in
bijection\footnote{Except for the term $\lambda f.\; f$, but it may be
  identified with $\overline{1} = \lambda f.\; (\lambda x.\; f\; x)$ by
  extending the equational theory with the innocuous \enquote{$\eta$-rule}: if
  $t : A \to B$ for some $A,B$ then $t =_\eta \lambda y.\; t\; y$.} with
$\naturalN$ via $n \mapsto \overline{n}$. So $\Nat$ can legitimately be seen as
the type of natural numbers in $\STlam$.

\subsubsection{A question: expressible functions in the simply typed
  $\lambda$-calculus}

At this point, we may ask: \emph{what are the functions $\naturalN \to
  \naturalN$ definable in $\STlam$?} As hinted in the introduction, this kind of
question depends heavily on the \emph{type} of the programs (i.e.\
$\lambda$-terms) that we use to code these functions. A classical result is:
\begin{theorem}[Schwichtenberg 1975~{\cite{Schwichtenberg}}]
  \label{thm:schwichtenberg}
  Let $f : \naturalN \to \naturalN$. There exists $t : \Nat \to \Nat$ such that
  $t\; \overline{n} =_\beta \overline{f(n)}$ for all $n \in \naturalN$ if and
  only if $f$ is an \emph{extended polynomial}, i.e.\ a function generated from
  0, 1, $+$, $\times$ and a conditional $\mathtt{if}\ n = 0\ \mathtt{then}\ p\
  \mathtt{else}\ q$.
\end{theorem}

So the $\lambda$-terms of type $\Nat \to \Nat$ have a rather low expressivity.
One trick to allow more functions to be defined is to perform a substitution of
the input type.
\begin{notation}
  For types $A$ and $B$, we abbreviate the substitution $A\{o:=B\}$ as $A[B]$.
\end{notation}

We shall consider $\lambda$-terms of type $\Nat[A] \to \Nat$ (by expanding the
definitions, $\Nat[A] = (A \to A) \to (A \to A)$) where $A$ is an arbitrary
simple type. Terms of this type still define numeric functions, thanks to a
simple \enquote{substitution lemma}: $\overline{n} : \Nat$ entails that
$\overline{n}$ can also be given the type $\Nat[A]$ for all types $A$ and all $n
\in \naturalN$. Typically, one can check that $(\lambda x.\; x\; \overline{2}) :
\Nat[o \to o] \to \Nat$ represents the function $n \mapsto 2^n$.

To our knowledge, there is only one characterization of the class of functions
$\naturalN \to \naturalN$ thus obtained, due to Joly~\cite{Joly}. It is
formulated in terms of untyped $\lambda$-terms subject to a kind of complexity
constraint in an unrealistic (by Joly's own admission) cost model. Therefore, it
would be of obvious interest to describe this class without reference to the
$\lambda$-calculus.

\begin{question}
  Characterize the functions $f : \naturalN \to \naturalN$ definable in $\STlam$
  in the following way: there exists a type $A$ and a term $t : \Nat[A] \to
  \Nat$ such that for all $n \in \naturalN$, $t\;\overline{n} =_\beta
  \overline{f(n)}$.
\end{question}

It might seem surprising that this problem is still open despite the central
role that the simply typed $\lambda$-calculus has played in programming language
theory and in proof theory for the past few decades. We believe that this is due
in part by some well-known facts (cf.~\cite{FortuneLeivant}) that suggest that
there might be no satisfying answer: while any tower of exponentials $2
\uparrow^h n$ of fixed height $h$ can be expressed by a term of type $\Nat[A_h]
\to \Nat$ ($A_h$ becoming increasingly complicated as $h \to +\infty$), many
simple functions of tame growth are inexpressible. If we look at functions of
two variables, there is a striking example: \emph{subtraction} cannot be defined
by any term of type\footnote{We see a function $f : \naturalN \times \naturalN
  \to \naturalN$ as the function $x \in \naturalN \mapsto (y \mapsto f(x,y)) \in
  \naturalN^\naturalN$; cf.\ \S\ref{sec:stlam}.} $\Nat[A] \to (\Nat[B] \to
\Nat)$, no matter what simple types $A,B$ are chosen.

One aim of the present paper -- starting with the subsection below, which lends
a new significance to old results -- is to argue that this pessimism is perhaps
unwarranted.

\subsubsection{The relevance of automata to $\lambda$-definability}
\label{sec:hillebrand}

To gain some insight on this problem, let us both generalize and (temporarily)
specialize it:
\begin{itemize}
\item We replace natural numbers by \emph{strings} over a finite alphabet
  $\Sigma$. There exists a simple type $\Str_\Sigma$ of \emph{Church-encoded
    strings} and an encoding $t \in \Sigma^* \rightsquigarrow \overline{t} :
  \Str_\Sigma$ inducing a bijection $\Sigma^* \cong (\{t \mid t :
  \Str_\Sigma\}/=_\beta)$. We recover Church numerals as the special case of
  Church-encoded strings over \emph{unary alphabets}: $\Nat = \Str_{\{1\}}$.
  Schwichtenberg's result on $\Nat \to \Nat$ (\cref{thm:schwichtenberg}) can be
  suitably generalized to $\Str_\Gamma \to \Str_\Sigma$
  (see~\cite{Zaionc,LeivantFA}).
\item We shall start by looking at predicates $\Str[A] \to \Bool$ -- i.e.\ at
  \emph{languages} -- instead of functions $\Str[A] \to \Str$, with the usual
  type  $\Bool = o \to (o \to o)$ of booleans in $\STlam$.
\end{itemize}
Fortunately, the languages definable in $\STlam$ are known:

\begin{theorem}[Hillebrand \& Kanellakis 1995 {\cite{HillebrandKanellakis}}]
  A language $L \subseteq \Sigma^*$ can be expressed as $L = \Lang(t) = \{w \in
  \Sigma^* \mid t\;\overline{w} =_\beta \mathtt{true} \}$ for some
  $\lambda$-term $t : \Str_\Sigma[A] \to \Bool$ and some simple type $A$ if and
  only if it is a \emph{regular language}.
\end{theorem}

Furthermore, Joly stated his result~\cite{Joly} for arbitrary free algebras, and
the above theorem also generalizes to a characterization of \emph{regular tree
  languages} for such free algebras (using the right definition of Church
encoding). As for the specialization to $\naturalN$, it tells us that a subset
of $\naturalN$ can be decided by a term of type $\Nat[A] \to \Bool$ if and only
if it is \emph{ultimately periodic} -- this fact is generalized in Joly's paper
to ultimately periodic subsets of $\naturalN^k$.

We deduce from the above theorem the regularity preservation claimed in
\cref{thm:stlam-intro}:
\begin{definition}
  \label{def:lambda-definable}
  A \emph{$\lambda$-definable string function} is a function $f : \Gamma^* \to
  \Sigma^*$ that can be expressed by some term $t : \Str_\Gamma[A] \to
  \Str_\Sigma$, in the sense that $t\;\overline{w} = \overline{f(w)}$ for all $w
  \in \Sigma^*$.
\end{definition}
\begin{corollary}
  The preimage of a regular language by a $\lambda$-definable function is regular.
\end{corollary}
\begin{proof}
  Let $f : \Gamma^* \to \Sigma^*$ be defined by some term $t : \Str_\Gamma[A]
  \to \Str_\Sigma$, and let $L \subseteq \Sigma^*$ be a regular language. Then
  $L = \Lang(u)$ for some $u : \Str_\Sigma[B] \to \Bool$. By the substitution
  lemma, $t$ can be given the type $\Str_\Gamma[A[B]] \to \Str_\Sigma[B]$, so
  one can define the term $\lambda x.\; u\;(t\;x) : \Str[A[B]] \to \Bool$ in
  $\STlam$. To conclude, observe that $f^{-1}(L) = \Lang(\lambda x.\;
  u\;(t\;x))$.
\end{proof}

The same substitution lemma can be used to establish that the
$\lambda$-definable string functions are closed under composition. To prove
\cref{thm:stlam-intro}, it remains only to show that HDT0L transductions are
$\lambda$-definable -- which is the subject of \cref{sec:stlam-string}.

A final word about the relevance of the HDT0L+composition class to numeric
functions (i.e.\ the unary case). We have already mentioned that Sénizergues
claims (without giving a proof) this class to be equivalent to his
\emph{$k$-computable mappings}~\cite{Senizergues}, defined in terms of a variant
of iterated pushdown automata. The unary version of these mappings, called the
\emph{$k$-computable sequences}, had been previously studied in detail by
Fratani and Sénizergues~\cite{Fratani}, who showed that they generalized some
integer sequences of interest in number theory. Thus, an optimistic scenario
could be: $\Nat[A] \to \Nat$ in $\STlam$, unary HDT0L+composition and
$k$-computable sequences all define the same class of functions $\naturalN \to
\naturalN$, making this class a canonical mathematical object.

Generalizing this to strings, we propose a concrete question related to our open
problem:
\begin{question}
  Are the $\lambda$-definable string functions of \cref{def:lambda-definable},
  the closure by composition of HDT0L transductions, and Sénizergues's
  $k$-computable mappings all the same class of functions from strings to
  strings?
\end{question}

\subsection{Other motivations and related work}
\label{sec:intro-other}

\subparagraph{An analogy with machine-free complexity}

Beyond the very concrete goal stated above, the present work is an attempt to
transpose to the context of automata some ideas from \emph{implicit
  computational complexity} (ICC) -- a field whose aim is to characterize
complexity classes without reference to a particular machine model. Another
field which fits the description just given is \emph{descriptive complexity},
which establishes correspondences of the form \enquote{the predicates in the
  complexity class $\mathcal{C}$ are exactly those expressible in the logic
  $\mathcal{L}_{\mathcal{C}}$}. Its very successful automata-theoretic
counterpart is the use of \emph{Monadic Second-Order Logic} (MSO) over various
structures, ranging from finite words to graphs, infinite trees, ordinals\ldots
Concerning transductions in MSO, see~\cite{EngelfrietHoogeboom,polyregularMSO}.
In contrast, the methods of ICC -- e.g.\ term rewriting, function algebras, or
$\lambda$-calculi -- have a more computational flavor. To sum up:

\begin{center}
  \begin{tabular}{l|c|c|}
    & declarative programming & functional programming \\
    \hline
    complexity & Descriptive Complexity & Implicit Complexity (ICC) \\
    \hline
    automata & Monadic Second-Order Logic (MSO) & \textbf{this paper: Church encodings} \\
    \hline
  \end{tabular}
\end{center}

We should mention that there already exists some ICC-like work on transduction
classes, for example function algebras for regular functions using
combinators~\cite{regularcombinators,RTE,Daviaud}. The closest to ours is
perhaps Bojańczyk's characterization of polyregular functions by a variant of
$\STlam$~\cite{polyregular}, which we discuss in \S\ref{sec:polyseq}. The main
difference is that both these works use primitive data types for strings,
whereas we encode strings as higher-order functions (i.e.\ functions taking
functions as arguments) inside \enquote{purely logical} calculi.

\subparagraph{A few words about verification (and linear logic)}

The bottom row of the above table is also related to the field of \emph{formal
  verification}. For instance, MSO over infinite words -- whose decidability was
proved by Büchi using automata -- subsumes Linear Temporal Logic. The relevance
of Church encodings in typed $\lambda$-calculi has been demonstrated in the
context of \emph{higher-order model checking}, an active field of research
concerned with verifying functional programs: see Grellois's PhD
thesis~\cite{grellois} and references therein. By generalizing this use of
Church encodings, Melliès was led to introduce higher-order parity
automata~\cite{HOParity}. The introduction to~\cite{HOParity} is particularly
instructive: it proposes a \enquote{dictionary between automata theory and the
  simply typed $\lambda$-calculus} via Church encodings.

We should mention that linear logic plays an important role in some of Grellois
and Melliès's work (e.g.~\cite{GrelloisMellies}). For MSO over infinite words,
there is also a recent application of linear logic, namely Pradic and Riba's
approach to the synthesis problem~\cite{LMSO}.

\subparagraph{Implicit complexity in $\EAlam$}

Linear logic and its byproducts have also been used for ICC: one of the first
works of this kind is the characterization of elementary recursive functions in
Girard's Elementary Linear Logic (ELL)~\cite{girardELL}. ELL later inspired the
elementary affine $\lambda$-calculus~\cite{Benedetti} -- or rather a variant
that we baptized \emph{a posteriori} $\muEAlam$ in~\cite{ealreg} -- which
refines this by giving types of programs corresponding to each level of the
$k$-EXPTIME hierarchy:
\begin{theorem}[{Baillot et al.~\cite{Benedetti}}]
  In $\muEAlam$, a predicate can be decided by a term of type $\oc\Str_\Sigma
  \multimap \oc^{k+2}\Bool$ iff it is in \emph{$k$-EXPTIME}. In particular,
  $\oc\Str_\Sigma \multimap \oc\oc\Bool$ corresponds to \emph{P}.
\end{theorem}
We overload notation: here $\Str_\Sigma$ and $\Bool$ are respectively the
$\EAlam$ types of strings over $\Sigma$ and of booleans; they differ from the
$\STlam$ types of the same name. The unary connective `$\oc$' is the
\emph{exponential modality} of linear logic, which marks a duplicable resource,
and plays a role in controlling complexity in $\muEAlam$; $\multimap$ is the
linear function arrow.

We recently showed~\cite{ealreg} that by replacing $\muEAlam$ by $\EAlam$ (i.e.\
by removing \emph{type fixpoints} from $\muEAlam$) we get \emph{regular
  languages} instead of polynomial time for the case $k=0$. Just as we saw for
$\STlam$ in the previous section, that means that $\oc\Str_\Gamma \multimap
\oc\Str_\Sigma$ is a type of regularity-preserving functions, closed under
composition (whereas in $\muEAlam$, it corresponds to the class FP of polynomial
time functions, cf.~\cite{ealreg}). Another type for regular languages in
$\EAlam$ is $\Str_\Sigma \multimap \oc\Bool$, so functions of type $\Str_\Gamma
\multimap \Str_\Sigma$ are also of interest. We prove:
\begin{theorem}
  \label{thm:ealam}
  The terms of type $\Str_\Gamma \multimap \Str_\Sigma$ (resp.\ $\oc\Str_\Gamma
  \multimap \oc\Str_\Sigma$) $\EAlam$ can only express \emph{linear (resp.\
    polynomial) time} and \emph{regularity-preserving} functions; on the other
  hand:
  \begin{itemize}
  \item all \emph{regular functions} can be defined by $\EAlam$ terms of both
    types;
  \item furthermore, the functions expressible with $\oc\Str_\Gamma \multimap
    \oc\Str_\Sigma$ are closed under \emph{composition by substitution}
    (\cref{def:cbs}), and therefore may have $\Omega(n^k)$ growth for any $k\in
    \mathbb{N}$.
  \end{itemize}
\end{theorem}

\section{Several classes of transductions}
\label{sec:transducers}

This section recalls the HDT0L, regular, and polyregular transductions, and
introduces our new \enquote{composition by substitution} operation. Along the
way, we prove the inclusions
\[ \{\text{regular + comp.\ by subst.}\} \subseteq \{\text{polyregular
    functions}\} \subsetneq \{\text{HDT0L + composition}\} \] as we promised in
the introduction. Finally, we (almost) define the regular \emph{tree} functions.

A preliminary remark for this section on automata: recall that for a finite
alphabet $\Sigma$, the set of words over $\Sigma$, denoted by $\Sigma^*$, is the
free monoid over the set of generators $\Sigma$. Therefore, any function $\Sigma
\to M$ to a monoid $M$ uniquely extends to a morphism $\Sigma^* \to M$.

\subsection{Register transducers and HDT0L systems}
\label{sec:register-transducer}

Our first machine model for string-to-string functions has been mentioned in the
introduction: it is the non-linear version of streaming string transducers.
Basically, we enrich finite automata with some memory: a finite number of
\emph{string-valued registers}. At each transition, the contents of the
registers can be recombined by concatenation. After the input has been entirely
read, an output function is invoked to determine the final result from the
registers.

\begin{definition}
  A \emph{register transducer} over input and output alphabets $\Gamma$ and
  $\Sigma$ consists of:
  \begin{itemize}
  \item a finite set $Q$ of \emph{states}, with an initial state $q_I \in Q$
  \item a finite set $R$ of \emph{registers} (or variable names), disjoint from
    $\Gamma$ and $\Sigma$
  \item a transition function $\delta : Q \times \Gamma \to Q \times (R \to
    (\Sigma \cup R)^*)$ (i.e.\ $Q \times \Gamma \to Q \times ((\Sigma \cup
    R)^*)^R$)
  \item an output function $F : Q \to (\Sigma \cup R)^*$
  \end{itemize}
  A \emph{configuration} of this register transducer\footnote{We borrow the name
    from \url{https://www.mimuw.edu.pl/~bojan/papers/toolbox.pdf}.} is a pair
  $(q,s)$ with $q \in Q$ and $s : R \to \Sigma^*$. For $c \in \Gamma$, we write
  $(q,s) \longrightarrow_{c} (q',s')$ when $(q',u) = \delta(q,c)$ and $s' = s^*
  \circ u$, where $s^* : (\Sigma \cup R)^* \to \Sigma^*$ is the monoid morphism
  taking $a \in \Sigma$ to itself and $r \in R$ to $s(r)$. 

  The \emph{image} of a string $w = w_1 \ldots w_n \in \Gamma^*$ by this
  register transducer is $s^*(F(q))$ where $q \in Q$ and $s : R \to \Sigma^*$
  are uniquely determined by $(q_I,(x \in X \mapsto \varepsilon))
  \longrightarrow_{w_1} \ldots \longrightarrow_{w_n} (q,s)$. This defines a
  function $\Gamma^* \to \Sigma^*$.
\end{definition}

For example, let $Q = \{q\}$, $R = \{X,Y\}$ and $\delta(q,c) = (q,(X \mapsto
Xc,\, Y \mapsto cY))$ for all $c \in \Gamma = \Sigma = \{a,b\}$. Then for $w =
w_1 \ldots w_n \in \Gamma^*$,
\[ (q, (x \in X \mapsto \varepsilon)) \longrightarrow_{w_1} \ldots
  \longrightarrow_{w_n} (q,s) \qquad \text{with}\ s(X) = w\ \text{and}\ s(Y) =
  \mathtt{reverse}(w) \]
So, if we take as output function $F(q) = XY$, the function defined is $w
\mapsto w \cdot \mathtt{reverse}(w)$.

Alternatively, these functions can be specified using monoid morphisms:

\begin{definition}
  A \emph{HDT0L system} consists of:
  \begin{itemize}
  \item an input alphabet $\Gamma$, an output alphabet $\Sigma$, and a working
    alphabet $\Delta$;
  \item an initial word $d \in \Delta^*$;
  \item for each $c \in \Gamma$, a \emph{monoid morphism} $h_c : \Delta^* \to
    \Delta^*$;
  \item a final morphism $h' : \Delta^* \to \Sigma^*$.
  \end{itemize}
  It defines the transduction taking $w = w_1 \ldots w_n \in \Gamma^*$ to $h'
  \circ h_{w_1} \circ \ldots \circ h_{w_n}(d) \in \Sigma^*$.
\end{definition}
The family $(h_c)_{c \in \Gamma}$ may be equivalently given as a morphism $H :
\Gamma^* \to \mathrm{Hom}(\Delta^*,\Delta^*)$ (the latter is a monoid for
function composition); the image of the word $w$ is then $h'((H(w))(d))$.

\begin{theorem}[Filiot \& Reynier~\cite{FiliotReynier}]
  A string function $\Gamma^* \to \Sigma^*$ can be computed by a register
  transducer iff it can be specified by a HDT0L system.
\end{theorem}

\subsection{(Poly)regular functions vs HDT0L(+composition)}

We shall take \emph{linear} register transducers as our definition of regular
functions. Enriching this class with a \enquote{squaring with underlining}
operation yields Bojańczyk's polyregular functions.

\begin{definition}[{Alur \& Černý~\cite{SST}}]
  A \emph{streaming string transducer (SST)} is a register transducer satisfying
  the \emph{copyless assignment} conditions: for all $r \in R$,
  \begin{itemize}
  \item for any register update in the transducer -- i.e.\ any $u : R \to
    (\Sigma \cup R)^*$ such that $(q',u) = \delta(q,c)$ for some $q,q' \in Q$
    and $c \in \Gamma$ -- $r$ appears at most once among all $u(r')$
    for $r' \in R$;
  \item for all $q \in Q$, $r$ appears at most once in the string $F(q)$.
  \end{itemize}
  A function $\Gamma \to \Sigma^*$ is \emph{regular} if it is computed by some
  SST.
\end{definition}

\begin{remark}
  The important part is the first item; the condition on output functions can be
  removed without increasing the expressivity of streaming string transducers.
\end{remark}

\begin{definition}
  Let $\Gamma$ be a finite alphabet. We write $\underline\Gamma =
  \{\underline{c} \mid c \in \Gamma\}$ for a disjoint copy of $\Gamma$ made of
  \enquote{underlined} letters. The function $\mathtt{squaring}_\Gamma :
  \Gamma^* \to (\Gamma \cup \underline\Gamma)^*$ is illustrated by the following
  example for $\Gamma = \{\mathtt{1},\mathtt{2},\mathtt{3},\mathtt{4}\}$:
  $\mathtt{squaring}_\Gamma(\mathtt{1234}) =
  \mathtt{\underline12341\underline23412\underline34123\underline4}$.
\end{definition}

\begin{definition}
  The class of \emph{polyregular functions} is the smallest class closed under
  composition containing the regular functions and the functions
  $\mathtt{squaring}_\Gamma$ for all finite $\Gamma$.
\end{definition}

Bojańczyk's original definition~\cite{polyregular} is the closure by composition
of \emph{sequential} functions, squaring and an additional \enquote{iterated
  reverse} function. Ours is equivalent because all regular functions are
polyregular, all sequential functions are regular, and iterated reverse is a
regular function (for this last point, we invite the reader to consult the
definition of iterated reverse in~\cite{polyregular} and check that a SST with
two registers suffices to compute it).

Let us now compare these classes to the HDT0L transductions (+ composition).

\begin{proposition}
  All regular functions can be specified by HDT0L systems.
\end{proposition}
\begin{proof}
  Streaming string transducers are special cases of register transducers.
\end{proof}

\begin{theorem}
  \label{thm:polyreg-hdt0l}
  All polyregular functions are compositions of HDT0L transductions.
\end{theorem}
\begin{proof}[Proof sketch]
  Thanks to the previous proposition, it suffices to show that the functions
  $\mathtt{squaring}_\Gamma$ can be computed by composing register transducers.
  We decompose them as
  \[ \mathtt{1234} \mapsto
    \mathtt{\underline{1}1\underline{2}12\underline{3}123\underline{4}} \mapsto
    \mathtt{\underline{4}321\underline{3}21\underline{2}1\underline{1}} \mapsto
    \mathtt{\underline{4}321(4)\underline{3}21(43)\underline{2}1(432)\underline{1}}
    \mapsto
    \mathtt{\underline{1}2341\underline{2}3412\underline{3}4123\underline{4}}
  \]
  where the parentheses are not part of the string, they only serve to help
  readability.
  \begin{itemize}
  \item The 1st step uses two registers, one for the output and one containing
    the current prefix.
  \item The 3rd step uses one register for the output and another register
    keeping track of the underlined characters seen thus far, by concatenating
    their non-underlined counterparts.
  \item The 2nd and 4th steps just apply the reverse function, which is
    regular.\qedhere
  \end{itemize}
\end{proof}

\begin{proposition}
  There exists a HDT0L transduction which is not polyregular.
\end{proposition}
\begin{proof}
  Polyregular functions have polynomial growth~\cite{polyregular}, while HDT0L
  transductions may grow exponentially. Take e.g.\ a HDT0L system with $h_a(b) =
  bb$ for all $a \in \Gamma, b \in \Sigma$.
\end{proof}

\subsection{Composition by substitutions vs polynomial list functions}
\label{sec:polyseq}

We come to our new operation on functions which allows increasing the exponent
of polynomial growth. It preserves polyregular functions, but this is not easy
to establish from the definition using the squaring function. We shall instead
rely on another characterization: an enriched variant of the simply typed
$\lambda$-calculus, called the \enquote{polynomial list functions} formalism
in~\cite{polyregular}.

\begin{definition}
  \label{def:cbs}
  Let $f : \Gamma^* \to I^*$, and for each $i \in I$, let $g_i : \Gamma^* \to
  \Sigma^*$. The \emph{composition by substitutions} of $f$ with the family
  $(g_i)_{i\in I}$ is the function
  \[\CbS(f,(g_i)_{i\in I}) : w \mapsto
  g_{i_1}(w) \ldots g_{i_n}(w)\ \text{where}\ f(w) = i_1 \ldots i_k\]
  That is, we first apply $f$ to the input, then every letter $i$ in the result
  of $f$ is substituted by the image of the original input by $g_i$.
  Thus, $\CbS(f,(g_i)_{i\in I})$ is a function $\Gamma^* \to \Sigma^*$.
\end{definition}

As an example, this can be used to define the \enquote{squaring without
  underlining} function\footnote{It is a classic exercise in formal languages to
  prove that if $L$ is a regular language, then $\{w \mid w^{|w|} \in L\}$ is
  also regular. Our study of superlinear transduction classes provides a wider
  context for this fact.} $w \mapsto w^{|w|}$, which can be expressed as $\CbS(f
: w \mapsto a^{|w|}, (g_a : w \mapsto w))$ with $f$ and $g_a$ regular. Its
growth rate is quadratic, while regular functions have at most linear growth.

\begin{remark}
  More generally, the smallest class containing regular functions and closed by
  both $\CbS$ and usual function composition contains, for all $k \in
  \naturalN$, some $f$ with $|f(w)| = \Theta(|w|^k)$. However, we conjecture
  that $\mathtt{squaring}_{\{1\}}$ (with underlining) is not in this class.
\end{remark}

We now recall how polynomial list functions are defined. They enrich the grammar
of $\lambda$-terms with constants whose meaning can be specified by extending
the $\beta$-rule of \S\ref{sec:intro-lambda-definable}, e.g.\
\[ \mathtt{is}_a\;b =_\beta \mathtt{true}\quad\text{if $a=b$}\qquad
  \mathtt{is}_a\;b =_\beta \mathtt{false}\quad\text{if $a\neq b$}\]
The grammar of types is also extended accordingly. For instance, any finite set
$\tau$ induces a type also written $\tau$, such that the elements $a \in \tau$
correspond to the terms $a : \tau$ of this type. There are also operations
expressing the cartesian product ($\times$) and disjoint union ($+$) of two
types; and a type of lists ($A^*$ is the type of lists over the type $A$). So
we actually consider
\[ \mathtt{is}_a^\tau : \tau \to
  \{\mathtt{true}\}+\{\mathtt{false}\}\qquad\text{for any finite set $\tau$} \]
and in the expression $\mathtt{is}^\tau_a\;b$, one therefore requires $b$ to be
part of a finite set $\tau$ specified in advance which also contains $a$.
See~\cite[Section~4]{polyregular} for the other primitive operations that are
added to $\STlam$; we make use of $\mathtt{is}$, $\mathtt{case}$, $\mathtt{map}$
and $\mathtt{concat}$ here. Bojańczyk's result is that if $\Gamma$ and $\Sigma$
are finite sets, then the polynomial list functions of type $\Gamma^* \to
\Sigma^*$ correspond exactly the polyregular functions.

\begin{remark}
  There is no substitution in the input type, and this is why our
  $\lambda$-definable string functions are still more expressive than polynomial
  list functions. On the other hand, this shows that primitive data types
  provide an alternative way of going beyond the poor expressive power
  (cf.~\cite{Zaionc,LeivantFA}) of the functions defined by $\Str_\Gamma \to
  \Str_\Sigma$ in $\STlam$.
\end{remark}

\begin{lemma}
  Let $I = \{i_1, \ldots, i_{|I|}\}$. Then the function $\mathtt{match}^{I,\tau}
  : I \to \tau \to \ldots \to \tau \to \tau$ which returns its $(k+1)$-th
  argument\footnote{See the beginning of \S\ref{sec:stlam} for an explanation of
    functions with multiple arguments in $\STlam$.} when its 1st argument is
  $i_k$ is a polynomial list function.
\end{lemma}
\begin{proof}[Proof sketch]
  By induction on $|I|$, it is definable from $\mathtt{is}^I_i$ ($i \in I$) \&
  $\mathtt{case}^{\{\mathtt{true}\},\{\mathtt{false}\},\tau}$.
\end{proof}

\begin{theorem}
  Polyregular functions are closed under composition by substitutions.
\end{theorem}
\begin{proof}
  Let $f : \Gamma^* \to I^*$, and for $i \in I$, $g_i : \Gamma^* \to \Sigma^*$
  be polyregular functions. Assuming that $f$ and $g_i$ ($i \in I$) are defined
  by polynomial list functions of the same name, $\CbS(f,(g_i)_{i\in I})$ can be
  expressed as $\lambda w.\;
  \mathtt{concat}^\Sigma\;(\mathtt{map}^{I,\Sigma^*}\;(\lambda i.\;
  \mathtt{match}^{I,\Sigma^*}\; i\; (g_{i_1}\; w)\;
  \ldots\;(g_{i_{|I|}}\;w))\;(f\;w))$.
\end{proof}

\subsection{Register tree transducers}

To define the regular tree functions, the first step is to consider the tree
version of register transducers. We shall restrict ourselves to binary trees, as
in~\cite[\S3.7]{BRTT}.
\begin{definition}
  The set $\BinTree(\Sigma)$ of \emph{binary trees} over the alphabet $\Sigma$,
  and the set $\dBinTree(\Sigma)$ of \emph{one-hole binary trees}\footnote{Our
    choice of notation is motivated by the fact that in enumerative
    combinatorics, the derivative of a generating function or species of
    structures corresponds to taking one-hole contexts.}, are generated by the
  respective grammars
  \[ T,U ::= \langle \rangle \mid a\langle T,U \rangle\quad (a \in \Sigma)
    \qquad\qquad T' ::= \square \mid a\langle T',T \rangle \mid a\langle T,T'
    \rangle \quad a \in \Sigma) \]

  That is, $\BinTree(\Sigma)$ consists of binary trees whose leaves are all
  equal to $\langle \rangle$ and whose nodes are labeled with letters in
  $\Sigma$. As for $\dBinTree(\Sigma)$, it contains trees with exactly one leaf
  labeled $\square$ instead of $\langle \rangle$. This \enquote{hole} $\square$
  is intended to be substituted by a tree: for $T' \in \dBinTree(\Sigma)$ and $U
  \in \BinTree(\Sigma)$, $T'[U]$ denotes $T'$ where $\square$ has been replaced
  by $U$.
\end{definition}

\begin{definition}
  The \emph{binary tree (resp.\ one-hole binary tree) expressions} over the
  variable sets $V$ and $V'$ are generated by the grammar (with $x \in V$, $x'
  \in V'$ and $a \in \Sigma$)
  \[ E,F ::= \langle \rangle \mid x \mid a\langle E,F \rangle \mid E'[E] \qquad
    \emph{(resp.}\ E',F' := \square \mid x' \mid a\langle E',E \rangle \mid
    a\langle E,E' \rangle \mid E'[F']\emph{)} \]
  The sets of such expressions is denoted by \emph{$\ExprBT(\Sigma,V,V')$
    (resp.\ $\ExprdBT(\Sigma,V,V')$)}.

  Given $\rho : V \to \BinTree(\Sigma)$ and $\rho' : V' \to \dBinTree(\Sigma)$,
  one defines $E(\rho,\rho') \in \BinTree(\Sigma)$ for $E \in \ExprBT(\Sigma)$
  and $E'(\rho,\rho') \in \dBinTree(\Sigma)$ for $E \in \ExprdBT(\Sigma)$ in the
  obvious way.
\end{definition}

\begin{definition}
  A \emph{register tree transducer (RTT)} $\Gamma^* \to \Sigma^*$ consists of: a
  finite set $Q$ of \emph{states} with an initial state $q_I \in Q$; two
  disjoint finite sets $R,R'$ of \emph{registers}; an output function $F : Q \to
  \ExprBT(\Sigma,R,R')$; and a transition function (where $R_{\ttl\ttr} =
  R\times\{\ttl,\ttr\}$)
  \[\delta : Q \times Q \times \Gamma \to Q \times (R \to
    \ExprBT(\Sigma,R_{\ttl\ttr},R'_{\ttl\ttr})) \times (R' \to
    \ExprdBT(\Sigma,R_{\ttl\ttr},R'_{\ttl\ttr})) \]
\end{definition}

The set of configurations of a RTT is $Q \times \BinTree(\Sigma)^R \times
\dBinTree(\Sigma)^{R'}$. It processes its input tree in a single bottom-up
traversal, computing for each subtree a configuration, starting with $(q_I, (r
\mapsto \langle \rangle), (r' \mapsto \square))$ at the leaves. The
configuration at $a\langle T,U \rangle$ is obtained from the one at $T$ and the
one at $U$ by applying $\delta$ to the pair of states and to $a$, and using each
expression $E$ in the image to determine the value $E(\rho,\rho')$ of the
corresponding register, where $\rho$ maps $(r,\ttl)$ (resp.\
$(r,\ttr)$) to the value of $r$ in the configuration of the left subtree
$T$ (resp.\ right subtree $U$), and similarly for $\rho'$.
See~\cite[\S3.7]{BRTT} for a more precise definition.

Regular tree functions are actually characterized by Alur and
D'Antoni's~\emph{bottom-up ranked tree transducers}~\cite{BRTT}. They are
register tree transducers with a kind of linearity condition, whose statement is
more complicated than in the case of SSTs. We give the full definition -- which
involves a \enquote{conflict relation} over registers -- in
\cref{sec:appendix-ealam-tree}.

\section{Transductions in the simply typed $\lambda$-calculus}
\label{sec:stlam}

\subsection{HDT0L+composition string functions are $\lambda$-definable}
\label{sec:stlam-string}

After these long preliminaries, at last, it is time to encode transductions in
$\STlam$.

First, we need to state precisely the definitions of Church encodings beyond
$\Nat$. For a finite alphabet $\Sigma$, we take $\Str_\Sigma = (o \to
o)^{|\Sigma|} \to o \to o$. This requires some explanations:
\begin{itemize}
\item The function arrow is left-associative, so this is the same as $(o \to
  o)^{|\Sigma|} \to (o \to o)$. In general, a term of type $A_1 \to \ldots \to
  A_n \to B = A_1 \to ( \ldots (A_{n-1} \to (A_n \to B)) \ldots)$ should be
  thought of as a function with $n$ inputs of type $A_1,\ldots,A_n$ and one
  output of type $B$ (this is analogous to the set-theoretic isomorphism $B^{A_1
    \times A_2} \cong (B^{A_2})^{A_1}$).
\item For the same reasons we abbreviate $A \to \ldots \to A \to B$ with $n$
  times $A$ as $A^n \to B$ (and at the level of terms, $(\ldots((f\; x_1)\; x_2)
  \ldots)\;x_n$ as $f\;x_1\;\ldots\;x_n$).
\end{itemize}
Observe that $\Str_{\{1\}} = (o \to o) \to o \to o = \Nat$ as we claimed in the
introduction.

Given an enumeration $\Sigma = \{a_1, \ldots, a_{|\Sigma|}\}$, a string $w =
a_{i_1} \ldots a_{i_n} \in \Sigma^*$ is encoded as
\[ \overline{w} = \lambda f_1.\; \ldots\; \lambda f_{|\Sigma|}.\; \lambda x.\;
  f_{i_1}\;(\ldots\; (f_{i_n}\;x)\ldots)\quad \text{(morally},\;
  \overline{w}(f_1,\ldots,f_{|\Sigma|}) = f_{i_1} \circ \ldots \circ
  f_{i_n}\text{)} \]
With this, the definition of $\lambda$-definable string functions
(\cref{def:lambda-definable}) is now fully rigorous.

The first two items in the statement of \cref{thm:stlam-intro} have already been
established in the introduction (more precisely \S\ref{sec:hillebrand}), so let
us prove the last one.

\begin{lemma}
  \label{lem:stlam-morphism}
  Any monoid morphism $\Gamma^* \to \Sigma^*$ can be defined by a term in
  $\STlam$ of type $\Str_\Gamma \to \Str_\Sigma$ -- there is no need for a
  substitution in the input type.
\end{lemma}
\begin{proof}
  For $\Gamma = \{g_1,\ldots,g_k\}$, $\Sigma = \{s_1, \ldots, s_l\}$, let
  $\varphi : \Gamma^* \to \Sigma^*$ be a morphism. We define
  \[ t = \lambda z.\; \left(\lambda f_1.\; \ldots\;\lambda f_l.\; z\;
      \left(\overline{\varphi(g_1)}\;f_1\;\ldots\;f_l\right) \;\ldots\;
      \left(\overline{\varphi(g_k)}\;f_1\;\ldots\;f_l\right)\right) \]
  One can check that $t : \Str_\Gamma \to \Str_\Sigma$ represents $\varphi$.
  Morally, the reason is that for $u = g_{i_1} \ldots g_{i_n} \in \Gamma^*$,
  $\overline{\varphi(g_{i_1})}(f_1,\ldots,f_n) \circ\ldots\circ
  \overline{\varphi(g_n)}(f_1,\ldots,f_n) = \overline{\varphi(g_{i_1}) \ldots
    \varphi(g_{i_n})}(f_1,\ldots,f_n)$.
\end{proof}

\begin{theorem}
  HDT0L transductions are $\lambda$-definable string functions.
\end{theorem}
\begin{proof}
  Consider a HDT0L system defining a function $f : \Gamma^* \to \Sigma^*$ with
  alphabets $\Gamma = \{g_1,\ldots,g_k\}$, $\Sigma$ and $\Delta$, initial word
  $d \in \Delta^*$, morphisms $h_i : \Delta^* \to \Delta^*$ for $i \in
  \{1,\ldots,k\}$ corresponding to $c_i \in \Gamma$, and final morphism $h' :
  \Delta^* \to \Sigma^*$. By the above lemma, each $h_i$ (resp.\ $h'$) can be
  represented by $u_i : \Str_\Delta \to \Str_\Delta$ (resp.\ $u' : \Str_\Delta
  \to \Str_\Sigma$).

  It is important to note that the input and output types of the $u_i$ are
  equal. This allows us to define the term $t = \lambda z.\; u'\;(z\; u_1\;
  \ldots\; u_k \; \overline{d}) : \Str_\Gamma[\Str_\Delta] \to \Str_\Sigma$
  which expresses $f$.
\end{proof}

\subsection{Regular tree functions are $\lambda$-definable}
\label{sec:stlam-tree}

In the case of tree-to-tree functions, we also prove that register tree
transducers (RTTs) can be encoded in $\STlam$ -- and consequently, their closure
under composition also can. However, we are not aware of any alternative
characterization of this class -- we only know that it it is a (strict)
superclass of the regular tree functions. So we must work directly with RTTs.

The type of \emph{Church encodings of binary trees} is $\BT_\Sigma = (o \to o
\to o)^{|\Sigma|} \to o \to o$ (where $\Sigma$ is the alphabet of node labels).
Given an enumeration $\Sigma = \{a_1,\ldots,a_n\}$, each $T \in
\BinTree(\Sigma)$ is encoded as a $\lambda$-term $\overline{T} = \lambda f_1.\;
\ldots\; \lambda f_n.\; \lambda x.\; \widehat{T} : \BT_\Sigma$, where we define
inductively $\widehat{\langle \rangle} = x$ and $\widehat{a_i \langle T,U
  \rangle} = f_i\; \widehat{T}\;\widehat{U}$. Morally, $\overline{T}(f_1,
\ldots, f_n, x)$ is the result of a single-pass bottom-up traversal of $T$,
starting with $x$ at the leaves and combining the results of subtrees with $T$.
\begin{remark}
  Analogously, $\Str_\Sigma$ can be seen as an encoding of \enquote{unary trees}
  whose bottom-up traversals correspond to right-to-left traversals of the
  corresponding strings (think of the \texttt{fold\_right} / \texttt{foldr}
  functions in some functional programming languages).
\end{remark}

\begin{lemma}
  \label{lem:stlam-rtt}
  Any $T' \in \dBinTree(\Sigma)$ can be compiled to a term $\Compile(T')$ in
  $\STlam$ of type $\dBT_\Sigma = \BT_\Sigma \to \BT_\Sigma$ such that for all
  $U \in \BinTree(\Sigma)$, $\Compile(T')\;\overline{U} =_\beta
  \overline{T'[U]}$. Similarly:
  \begin{itemize}
  \item any $E \in \ExprBT(\Sigma, V = \{x_1, \ldots, x_n\}, V' =
    \{x'_1,\ldots,x'_m\})$ can be compiled to a $\lambda$-term $\Compile(E) :
    (\BT_\Sigma)^n \to (\dBT_\Sigma)^m \to \dBT_\Sigma \to \BT_\Sigma$ such that
    for all $\rho : V \to \BinTree(\Sigma)$ and $\rho' : V' \to
    \dBinTree(\Sigma)$, $\overline{E(\rho,\rho')} =_\beta
    \Compile(E)\;\overline{\rho(x_1)}\;\ldots\;\overline{\rho(x_n)}\;
    \Compile(\rho'(x'_1))\;\ldots\;\Compile(\rho'(x'_m))$.
  \item any $E' \in \ExprdBT(\Sigma)$ can be compiled to a term $\Compile(E') :
    (\BT_\Sigma)^n \to (\dBT_\Sigma)^m \to \dBT_\Sigma$ enjoying the analogous
    property.
  \end{itemize}
\end{lemma}

\begin{theorem}
  \label{thm:stlam-rtt}
  Any function from $\BinTree(\Gamma)$ to $\BinTree(\Sigma)$ computed by a
  register tree transducer can be expressed by a $\lambda$-term of type
  $\BT_\Gamma[A] \to \BT_\Sigma$ for some simple type $A$.
\end{theorem}
\begin{proof}[Proof sketch]
  As discussed above, the kind of bottom-up traversal done by a register tree
  transducer corresponds exactly to the \enquote{fold function} embodied by the
  Church encoding of a tree. One would want to directly encode the RTT by
  setting $A$ to be its type of configurations; the main obstacle to defining
  such an $A$ is the lack of product and sum types in $\STlam$ (unlike in
  polynomial list functions, cf.~\S\ref{sec:polyseq}). To overcome this, we use
  a continuation-passing-style transformation with return type $\BT_\Sigma$.
  Cf.~\cref{sec:appendix-stlam-rtt} for details of the proof.
\end{proof}

\begin{corollary}
  Any regular tree function is definable in $\STlam$.
\end{corollary}

\section{Streaming transducers in the elementary affine
  $\lambda$-calculus}
\label{sec:ealam}

The grammar of terms of $\EAlam$ and its equational theory are given by
\[ t, u ::= x \mid \lambda x.\, t \mid \lambda\oc x.\, t \mid t\,u \mid \oc t
  \qquad
 (\lambda x.\, t)\, u =_\beta t\{x := u\} \quad
  (\lambda\oc x.\, t)\, (\oc u) =_\beta t\{x := u\}
\]
where $x$ is taken in a countable set of variables; we take $=_\beta$ to be the
smallest congruence generated by the two rules above. The type system of
$\EAlam$ is given in \cref{sec:appendix-ealam-type-system}. It enforces two
important constraints on terms. The first means that one must use $\lambda\oc$
to define non-linear functions -- in other words, a subterm must be marked by
`$\oc$' to be duplicable:
\begin{center}
  \emph{(linearity)} in any subterm of the form $\lambda x.\; t$, $x$ appears
  \emph{at most once} in $t$
\end{center}
An additional constraint related specifically to Elementary Linear
Logic~\cite{girardELL} is
\begin{center}
  \emph{(stratification)} in any subterm of the form $\lambda x.\; t$ (resp.\
  $\lambda\oc x.\; t$),\\
  the \emph{depth} of each occurrence of $x$ in $t$ is 0 (resp.\ 1)
\end{center}
By depth we mean the number of $\oc$'s in $t$ surrounding $x$. Stratification
entails that in the two rules above generating $=_\beta$, the depth of the
subterm $u$ is the same on both sides; thus, we have an invariant for $=_\beta$.
In particular one cannot define type-cast functions taking any $\oc{t}$ to $t$
(\emph{dereliction}) or to $\oc\oc{t}$ (\emph{digging}). (`$\oc$' is called the
\emph{exponential modality}.)

\subsection{Encoding streaming string transducers}
\label{sec:ealam-string}

The Church-encoded strings over $\Sigma = \{a_1, \ldots, a_n\}$ are defined in
$\EAlam$ as:
\[ \text{for $w = a_{i_1} \ldots a_{i_n} \in \Sigma^*$},\quad \overline{w} =
  \lambda\oc f_1.\; \ldots \lambda\oc f_n.\; \oc(\lambda x.\; f_{i_1}\, (\ldots
  (f_{i_n}\;x)\ldots))\]
and they are given the type $\Str_\Sigma = \forall \alpha.\;
\Str_\Sigma[\alpha]$ where $\Str_\Sigma[\alpha] = (\oc(\alpha \multimap
\alpha))^{|\Sigma|} \multimap \oc(\alpha \multimap \alpha)$. (As we did for
$\STlam$, we abbreviate $A \multimap \ldots \multimap A \multimap B$ with $k$
times $A$ as $A^k \multimap B$.) The $\forall\alpha$ is a second-order
quantifier -- the type system of $\EAlam$ indeed supports polymorphism.

Another encoding in $\EAlam$ is that of the finite set $\{1,\ldots,k\}$,
represented by the type $\Fin(n) = \forall \alpha.\; \alpha^n \multimap \alpha$:
the encoding of $i \in \{1,\ldots,k\}$ is $\lambda x_1.\; \ldots\; \lambda
x_k.\; x_i$. For instance the type $\Bool$ mentioned in~\S\ref{sec:intro-other}
is $\Fin(2) = \forall \alpha.\; \alpha \multimap \alpha \multimap \alpha$ --
this mirrors the $\STlam$ booleans.

As we discussed in \S\ref{sec:stlam-tree}, it is most natural to process a
string right-to-left using its Church encoding. But register transducers work in
a left-to-right fashion. To compensate for that, we shall \emph{propagate output
  functions backwards} instead.

\begin{definition}
  \label{def:delta-o}
  Let $(Q,q_I,R,\delta,F)$ be a register transducer with input alphabet
  $\Gamma$. We define $\delta^O : \Gamma \times (Q \to (\Sigma \cup R)^*) \to (Q
  \to (\Sigma \cup R)^*)$ by $\delta^O(a,G) = (q \in Q \mapsto
  s_{a,q}^*(G(q'_{a,q})))$, where $(q'_{a,q},s_{a,q}) = \delta(q,a)$ and
  $s_{a,q}^*$ is the unique extension of $s_{a,q} : R \to (\Sigma \cup R)^*$ to
  a monoid morphism $(\Sigma \cup R)^* \to (\Sigma \cup R)^*$ taking each letter
  of $\Sigma$ to itself.
\end{definition}
\begin{proposition}
  \label{prop:delta-o}
  Let $w = w_1 \ldots w_n \in \Gamma^*$. The image of $w$ by the register
  transducer $(Q,q_I,R,\delta,F)$ is $\varphi(G(q_I))$ where $G = \delta^O(w_1,
  (\ldots \delta^O(w_n, F) \ldots))$ and $\varphi : (\Sigma \cup R)^* \to
  \Sigma^*$ erases all letters from $R$ in its input.
\end{proposition}

We must now implement this idea as a term of type $\Str_\Gamma \multimap
\Str_\Sigma$ in $\EAlam$. \emph{This is the same thing as a term of type}
$\Str_\Gamma[A] \multimap \Str_\Sigma[\alpha]$, where $\alpha$ is a free type
variable and $A$ may contain $\alpha$: by linearity, the quantified type
variable in the input is instantiated only once.

To implement this, one would want to iterate over the type of output functions;
naively, one would set $A$ to be $\Fin(|Q|) \multimap \Str_{\Sigma \cup R}$.
However this type contains an exponential (inside $\Str_{\Sigma \cup R}$) and
so, \emph{because of the stratification property, it is useless to produce an
  output of type $\Str[\alpha]$ since $\alpha$ is exponential-free}. (This can
be made rigorous using the \emph{truncation} operation for $\EAlam$ introduced
in~\cite{ealreg}.) Instead, we shall iterate over the \emph{purely linear} type
$A = \Fin(|Q|) \multimap (\alpha \multimap \alpha)^{|R|} \multimap (\alpha
\multimap \alpha)$.
It differs from the previous candidate by the absence of exponentials and of
$|\Sigma|$ arguments of type $\alpha \multimap \alpha$. This reflects the fact
that, if $F$ is a \emph{copyless} output function, then for all $q \in Q$,
$\overline{F(q)}$ is \emph{linear} in all arguments corresponding to register
names. As for those corresponding to $\Sigma$, they will be somehow replaced
with non-linear variables provided by the context.

We illustrate the construction on the register transducer computing $w \mapsto w
\cdot \mathtt{reverse}(w)$ given in \S\ref{sec:register-transducer}, which is
actually a streaming string transducer (that is, it is copyless). The general
proof is given in \cref{sec:appendix-ealam-regular}. We make a further
simplication: since this transducer has a single state, we drop the $\Fin(|Q|)$
argument in the type $A$. There are 2 registers, so our term has type
$\Str_{\{a,b\}}[A] \multimap \Str_{\{a,b\}}[\alpha]$ for $A = (\alpha \multimap
\alpha) \multimap (\alpha \multimap \alpha) \multimap (\alpha \multimap
\alpha)$.

First, we define $\EAlam$ terms corresponding to each $\delta^O(c,-)$
(\cref{def:delta-o}) for $c \in \{a,b\}$:
\[ d_c = \lambda G.\; \lambda r_X.\; \lambda r_Y.\; G\; (\lambda z.\; r_X\;
  (f_c\; z)) \; (\lambda z.\; f_c\; (r_Y\; y)) : A \multimap A \] These terms
use non-linearly the free variables $f_a, f_b : \alpha \multimap \alpha$.
Observe that the linearity condition of $\EAlam$ ($r_X$ and $r_Y$ occur at most
once) is satisfied precisely because the corresponding register update is
copyless! Next, we define $t : \Str_{\{a,b\}}[A] \multimap
\Str_{\{a,b\}}[\alpha]$ as
\[ t = \lambda u.\; \lambda\oc f_a.\; \lambda\oc f_b.\; (\lambda\oc h.\; \oc(h\;
  (\lambda r_X.\; \lambda r_Y.\; (\lambda z.\; r_X\; (r_Y\; z)))\; (\lambda
  x.\; x)\; (\lambda y.\; y)))\; (u\; \oc{d_a}\; \oc{d_b}) \]

Note that $\oc{d_a}$ contains $f_a$ at depth 1, bound by $\lambda\oc f_a$. Let
$w = w_1 \ldots w_n \in \{a,b\}^*$. Then $\overline{w}\; \oc{d_a}\; \oc{d_b}
=_\beta \oc(\lambda x.\; d_{w_1}\; (\ldots (d_{w_n}\; x) \ldots))$. Passing this
as argument to $(\lambda\oc h.\; \ldots)$ unpacks this exponential: $h = \lambda
x.\; d_{w_1}\; (\ldots (d_{w_n}\; x) \ldots)$. Next, $h$ is applied to a
representation of the output function $F(q) = XY$; so what we obtain represents
$\delta^O(w_1,\ldots,\delta^O(w_n,F))$. Indeed,
\[ (d_{w_1}\circ \ldots \circ d_{w_n})\; (\lambda r_X.\; \lambda r_Y.\; r_X\circ r_Y) =_\beta \lambda r_X.\; \lambda r_Y.\; r_X \circ f_{w_1}
  \ldots \circ f_{w_n} \circ f_{w_n} \circ \ldots \circ f_{w_1} \circ r_Y \]
where $g_1 \circ \ldots \circ g_m$ is an abbreviation for $\lambda x.\; g_1\;
(\ldots (g_m\; x)\ldots)$. By applying the above to two identity functions, we
erase $r_X$ and $r_Y$; thus, in the end, we get $t\;\overline{w} =_\beta
\overline{w \cdot \mathtt{reverse}(w)}$.

In general, since streaming string transducers can compute all regular
functions:
\begin{theorem}[proved in \cref{sec:appendix-ealam-regular}]
  \label{thm:ealam-string}
  Any regular function $\Gamma^* \to \Sigma^*$ can be computed by an $\EAlam$
  term of type $\Str_\Gamma \multimap \Str_\Sigma$ or $\oc\Str_\Gamma \multimap
  \oc\Str_\Sigma$.
\end{theorem}
The last part is because any term of type $A \multimap B$ in $\EAlam$ can be
type-cast into a term of type $\oc{A} \multimap
\oc{B}$~\cite[Proposition~28]{Benedetti}.

We have done the hard part in proving \cref{thm:ealam}. There remains only:
\begin{proposition}
  \label{prop:ealam-cbs}
  The expressible functions for the type $\oc\Str_\Gamma \multimap
  \oc\Str_\Sigma$ in $\EAlam$ are closed under composition by substitution.
\end{proposition}
\begin{proof}
  See \cref{sec:appendix-ealam-polyseq}.
\end{proof}

\begin{theorem}
  Any $\EAlam$ term of type $\Str_\Gamma \multimap \Str_\Sigma$ (resp.\
  $\oc\Str_\Gamma \multimap \oc\Str_\Sigma$) defines a function computable in
  linear (resp.\ polynomial) time.
\end{theorem}
\begin{proof}[Proof sketch]
  Let us start with $\oc\Str_\Gamma \multimap \oc\Str_\Sigma$. We proved
  in~\cite{ealreg} (building on work in~\cite{Benedetti}) that, in a larger
  system called $\muEAlam$, this type corresponds exactly to polynomial time
  functions. In particular, when we restrict to the subsystem $\EAlam$, the
  polynomial time upper bounds still hold. For $\Str_\Gamma \multimap
  \Str_\Sigma$, we can routinely adapt the arguments in~\cite{ealreg,Benedetti}
  to obtain a linear time bound for $\muEAlam$. The algorithm is to perform
  $\beta$-reduction with a particular \enquote{stratified} reduction strategy.
\end{proof}

\subsection{Bottom-up ranked tree transducers and
  the two linear conjunctions}

The $\EAlam$ type of Church-encoded binary trees with node labels in $\Sigma =
\{a_1,\ldots,a_{|\Sigma|}\}$ is
\[ \BT_\Sigma = \forall \alpha.\; \BT_\Sigma[\alpha] \qquad\text{where}\
  \BT_\Sigma[\alpha] = (\oc(\alpha \multimap \alpha \multimap
  \alpha))^{|\Sigma|} \multimap \oc\alpha \multimap \oc\alpha \]
To each $T \in \BinTree(\Sigma)$ we associate $\overline{T} : \BT_\Sigma$ in the
obvious way.
\begin{theorem}
  \label{thm:ealam-tree}
  Any regular tree function can be expressed by some $t : \BT_\Gamma \multimap
  \BT_\Sigma$ in $\EAlam$.
\end{theorem}
\begin{proof}[Proof sketch]
  We give only the main ideas here; a more detailed proof is provided in
  \cref{sec:appendix-ealam-tree}. As before, this amounts to translating any
  bottom-up ranked tree transducer (BRTT) to some term $t : \BT_\Gamma[A]
  \multimap \BT_\Sigma[\alpha]$, where $A$ may contain the type variable
  $\alpha$. Here, the natural direction of processing for a Church-encoded
  binary tree is bottom-up, and this coincides with the way a BRTT works, unlike
  the case of strings in the previous subsection.

  First, let us consider the case of a register tree transducer
  $(Q,q_I,R,R',F,\delta)$ enjoying a linearity condition directly analogous to
  streaming string transducers (SSTs). Then we take
  \[ A = \Fin(|Q|) \otimes \alpha^{\otimes|R|} \otimes (\alpha \multimap
    \alpha)^{\otimes|R'|} \qquad\text{representing \emph{configurations} of the
      BRTT}\ \]
  the use of $\otimes$ denoting the second-order encoding of the
  \emph{multiplicative conjunction}
  \[ A_1 \otimes \ldots \otimes A_m = \forall \beta.\; (A_1 \multimap \ldots
    \multimap A_m \multimap \beta) \multimap \beta \qquad B^{\otimes m} = B
    \otimes \ldots \otimes B \]

  An element of $\BinTree(\Sigma)$ (resp.\ $\dBinTree(\Sigma)$) contained in a
  register is therefore represented as a term of type $\alpha$ (resp.\ $\alpha
  \multimap \alpha$), using non-linearly the free variables $f_i : \alpha
  \multimap \alpha \multimap \alpha$ ($i \in \{1,\ldots,|\Sigma|\}$) and $x :
  \alpha$. To compare with the encoding of SSTs, a string which supports
  concatenation \emph{on both sides} can be seen as a one-hole unary tree, hence
  its type $\alpha \multimap \alpha$. (The uniqueness of the hole in
  $\dBinTree(\Sigma)$ turns out to be a linearity condition as well!) It is then
  possible to encode the transitions and output function of the BRTT.

  In general, a BRTT is a register transducer equipped with a reflexive and
  symmetric \emph{conflict relation} $\incoh$ over $R \cup R'$, and it satisfies
  a relaxed linearity condition formulated in terms of $\incoh$.
  Following~\cite{BRTT}, we say that $P \subseteq R \cup R'$ is
  \emph{non-conflicting} if $\forall x,y \in P,\, x = y \lor x \not\incoh y$. We
  take $A$ to be the following, \emph{where $P$ ranges over non-conflicting
    subsets}:
  \[ A = \Fin(|Q|) \otimes \bigwith_P \left( \alpha^{\otimes|P \cap R|} \otimes
      (\alpha \multimap \alpha)^{\otimes|P \cap R'|} \right) \]
  using the second-order encoding of the \emph{additive conjunction}
  \[ A_1 \with \ldots \with A_m := \forall \gamma.\; (\forall \beta.\; \beta
    \multimap (\beta \multimap A_1) \multimap \ldots \multimap (\beta \multimap
    A_m) \multimap \gamma)\multimap \gamma \]
  Further explanations of this choice and the role of $\incoh$ are given in
  \cref{sec:appendix-ealam-tree}.
\end{proof}

\section{Conclusion}

We exhibited some relationships between the functions between Church-encoded
strings (or trees) in two typed $\lambda$-calculi and those computed by variants
of finite-state transducers. On the automata-theoretic side, we showed that the
closure under composition of HDT0L transductions is a superclass of many
pre-existing transduction classes. By showing that this large transduction class
is included in the $\lambda$-definable string functions, we advanced our
understanding of the latter. As for $\EAlam$, the results here are still
preliminary; hopefully, the sequel to this paper should prove the converse
inclusions to Theorems~\ref{thm:ealam-string} and~\ref{thm:ealam-tree}, giving a
characterization of regular (tree) functions quite different from the already
existing ones.

Aside from that, there are many imaginable perspectives around the theme
\enquote{implicit complexity for automata}. For instance, is it possible to
characterize star-free languages in some $\lambda$-calculus, analogously to
their algebraic characterization by aperiodic monoids?

\bibliography{workshops/bi}

\appendix

\section{Register tree transducers in $\STlam$}
\label{sec:appendix-stlam-rtt}

This section is dedicated to the proof of \cref{thm:stlam-rtt}.

First, let us sketch the proof of \cref{lem:stlam-rtt}. Let $\Sigma = \{a_1,
\ldots, a_n\}$. We define $\Compile$ over $\dBinTree(\Sigma)$ by induction:
\begin{itemize}
\item $\Compile(\square) = \lambda z.\;z$,
\item $\Compile(a_i\langle T',U \rangle) = \lambda z.\; \lambda f_1.\; \ldots
  \; \lambda f_n.\; \lambda x.\;
  f_i\;(\Compile(T')\;z\;f_1\;\ldots\;f_n\;x)\;
  (\overline{U}\;f_1\;\ldots\;f_n\;x)$,
\item $\Compile(a_i\langle U, T' \rangle) = \lambda z.\; \lambda f_1.\; \ldots
  \; \lambda f_n.\; \lambda x.\; f_i\; (\overline{U}\;f_1\;\ldots\;f_n\;x)\;
  (\Compile(T')\;z\;f_1\;\ldots\;f_n\;x)$.
\end{itemize}
The compilation of expressions follows a similar scheme, with more cases. In
particular function application plays the main role in the translation of
$E[E']$ and $E'[F']$ to $\lambda$-terms.

Next, let $(Q,R,R',F,\delta)$ be a register tree transducer. We may assume
without loss of generality that $Q = \{1,\ldots,|Q|\}$. Our goal is to encode
this transducer into a simply typed $\lambda$-term of type $\BT_\Gamma[A] \to
\BT_\Sigma$. We take
\[ A = B^{|Q|} \to \BT_\Sigma \quad\text{where}\ B = \BT_\Sigma^{|R|} \to
  (\BT_\Sigma \to \BT_\Sigma)^{|R'|} \to \BT_\Sigma \] (Recall that $C^m \to D$
is merely an abbreviation for $C \to \ldots \to C \to D$.)

A down-to-earth explanation\footnote{For the reader familiar with programming
  language theory, a more conceptual explanation is that this type is isomorphic
  to $\lnot'\lnot'((1 + \ldots + 1) \times\BT_\Sigma^{|R|} \times (\BT_\Sigma
  \to \BT_\Sigma)^{|R'|})$, where $\lnot' D = D \to \BT_\Sigma$. This
  relativized double negation is used to eliminate the $\times$ and $+$ type
  constructors, which do not exist in our version of $\STlam$. As stated before,
  we are indeed using a continuation-passing-style transformation.} of these
types is as follows. The functions of $B$ take as input the contents of the
registers, and uses this to produce a result of type $\BT_\Sigma$. In
particular, recall that when the transducer has finished visiting the entire
tree, an \emph{output function} (depending on the final state) is called to
determine the result from the final contents of the registers; this function can
be expressed as a $\lambda$-term $u$ of type $B$.

As for $A$, the terms of type $A$ include (among others) all the terms of the form
(for $q \in Q$, $T_k \in \BinTree(\Sigma)$ and $T'_l \in \dBinTree(\Sigma)$)
\[ \mathtt{Conf}(q, (T_k)_{k \in R}, (T'_l)_{l \in R'}) = \lambda f_1.\;
  \ldots\; \lambda f_{|Q|}.\; f_q\; \overline{T_1} \;\ldots\; \overline{T_{|R|}}
  \; \Compile(T'_1) \;\ldots\; \Compile(T'_{|R'|}) \] Thanks to this, we can use
$A$ to represent $Q \times \BinTree(\Sigma)^R \times \dBinTree(\Sigma)^{R'}$,
that is, the set of \emph{configurations} of the register tree transducer
(assuming that we are in the middle of a computation whose final result will be
of type $\BT_\Sigma$). When, at some point, the transducer is at state $q \in Q
= \{1,\ldots,|Q|\}$, and its registers contain $(T_k)_{k \in R}$ and $(T'_l)_{l
  \in R'}$, the $\lambda$-term associated to its current configuration takes the
$q$-th input function and gives it as arguments these register contents. Of
course, the encoding depends of a fixed enumeration of the registers: $R =
\{\hat{r}_1,\ldots,\hat{r}_{|R|}\}$ and $R' =
\{\hat{r}'_1,\ldots,\hat{r}'_{|R'|}\}$.

The above discussion suggests that our register tree transducer be translated to
a $\lambda$-term of the following form, for some $L : A$ and
$N_1,\ldots,N_{|\Gamma|} : A \to A \to A$:
\[ \lambda z.\; (z\;N_1\;\ldots\;N_{|\Gamma|}\;L)\;u_1\;\ldots\;u_{|Q|} :
  \BT_\Gamma[A] \to \BT_\Sigma \]
where $u_q$ encodes the output function at the state $q \in Q$, in such a way
that for all $T \in \BinTree(\Sigma)$,
$\overline{T}\;L\;N_1\;\ldots\;N_{|\Gamma|} : A$ is (up to $=_\beta$) the
representation of the final configuration reached by the transducer when it
reads $T$.

The remaining task is to define $L$ and $N_1,\ldots,N_{|\Gamma|}$. Obviously $L$
should represent the initial configuration: writing $q_I$ for the initial state,
$L = \mathtt{Conf}(q_I,(\langle \rangle)_{k \in R},(\square)_{l \in R'})$.
Concerning $N_i$ for $i \in \{1,\ldots,|\Gamma|\}$, the property we want is that
\[ N_i\;\mathtt{Conf}(q_{\ttl}, (T_k)_{k \in R}, (T'_l)_{l \in
    R'})\;\mathtt{Conf}(q_{\ttr}, (U_k)_{k \in R}, (U'_l)_{l \in R'}) =_\beta
  \mathtt{Conf}(q, (V_k)_{k \in R}, (V'_l)_{l \in R'}) \]
for some $q, (V_k), (V'_l)$ determined by the variable-update and state-update
rules of the transducer for the $i$-th letter $g_i$ of $\Gamma = \{g_1, \ldots,
g_{|\Gamma|}\}$.

To define these configuration-update terms, we first define the terms
$M_{i,q_{\ttl},q_\ttr} : A$ containing the free variables
$r_{k,\ttl},r_{k,\ttr}$ of type $\BT_\Sigma$ for $k \in \{1,\ldots,R\}$, and
$r'_{l,\ttl},r'_{l,\ttr}$ of type $\BT_\Sigma \to \BT_\Sigma$ for $l \in R'$.
For $q_\ttl, q_\ttr \in Q$ and $g_i \in \Gamma$, if $\delta(q_\ttl,q_\ttr,g_i) =
(q,\psi,\psi')$ then
\[ M_{i,q_{\ttl},q_{\ttr}} = \lambda f_1.\; \ldots\; \lambda
  f_{|Q|}.\; f_q\;
  (\Compile(\psi(\hat{r}_1))\;r_{1,\ttl}\;\ldots\;r'_{|R'|,\ttr})\;
  \ldots\;
  (\Compile(\psi'(\hat{r}'_{|R'|}))\;r_{1,\ttl}\;\ldots\;r'_{|R'|,\ttr})\]
with $|R|+|R'|$ arguments passed to $f_q$.

Then the following choice for $N_i$ works: $N_i = \lambda c_{\ttl}.\; \lambda
c_{\ttr}.\; c_\ttl\; H_{i,1}\; \ldots \; H_{i,|Q|}$ where
\begin{itemize}
\item $H_{i,q_\ttl} = \lambda \vec{r_{\ttl}}\vec{r'_{\ttl}}.\; c_{\ttr}\;
  (\lambda \vec{r_{\ttr}}\vec{r'_{\ttr}}.\; M_{i,q_\ttl,1}) \;\ldots\; (\lambda
  \vec{r_{\ttr}}\vec{r'_{\ttr}}.\; M_{i,q_\ttl,|Q|})$;
\item for any term $t$, $\lambda \vec{r_{\ttl}}\vec{r'_{\ttl}}.\; t$ is an
  abbreviation for $\lambda r_{1,\ttl}.\;\ldots\;\lambda r_{|R|,\ttl}.\; \lambda
  r'_{1,\ttl}.\; \ldots\; \lambda r'_{|R'|,\ttl}.\; t$, and similarly for
  $\lambda \vec{r_{\ttr}}\vec{r'_{\ttr}}.\; t$.
\end{itemize}





\section{Details on transductions in $\EAlam$ (\cref{sec:ealam})}

\subsection{The type system of $\EAlam$}
\label{sec:appendix-ealam-type-system}

The following is mostly copied from our previous work~\cite{ealreg}.

The grammar of types for $\EAlam$ is
\[ A ::= \alpha \mid S \qquad S ::= \sigma \multimap \tau \mid \forall
  \alpha.\;S \qquad \sigma, \tau ::= A \mid \oc \sigma \]

The two first classes of types are called respectively \emph{linear} and
\emph{strictly linear}. (We follow the terminology of~\cite{Benedetti};
\enquote{linear} does not mean exponential-free, it merely means that the head
connective is not an exponential.)

The typing judgements involve a context split into three parts: they are of the
form $\Gamma \mid \Delta \mid \Theta \vdash t : \sigma$. The idea is that the
partial assignements $\Gamma$, $\Delta$ and $\Theta$ of variables to types
correspond respectively to linear, non-linear and \enquote{temporary} variables;
accordingly, $\Gamma$ maps variables to linear types (denoted $A$ above),
$\Delta$ maps variables to types of the form $\oc\sigma$, while $\Theta$ maps
variables to arbitrary types. The domains of $\Gamma$, $\Delta$ and $\Theta$ are
required to be pairwise disjoint. The derivation rules for $\EAlam$ are:
\[\text{variable rules}\qquad \frac{}{\Gamma, x : A \mid \Delta \mid \Theta
    \vdash x : A} \qquad \frac{}{\Gamma \mid \Delta \mid \Theta, x : \sigma \vdash x
    : \sigma}\]
\[\text{abstraction rules}\qquad \frac{\Gamma, x : A \mid \Delta \mid \Theta
    \vdash t : \tau}{\Gamma \mid \Delta \mid \Theta \vdash \lambda x.\; t : A
    \multimap \tau} \qquad \frac{\Gamma \mid \Delta, x : \oc\sigma \mid \Theta
    \vdash t : \tau}{\Gamma \mid \Delta \mid \Theta \vdash \lambda\oc x.\; t :
    \oc\sigma \multimap \tau} \]
\[\text{application rule\footnotemark}\qquad \frac{\Gamma \mid \Delta \mid \Theta
    \vdash t : \sigma \multimap \tau \quad \Gamma' \mid \Delta \mid \Theta
    \vdash u : \sigma}{\Gamma \uplus \Gamma' \mid \Delta \mid \Theta \vdash t\;u
    : \tau} \]
\footnotetext{$\Gamma \uplus \Gamma'$ means $\Gamma \cup \Gamma'$ with the
  assumption that the domains of $\Gamma$ and $\Gamma'$ are disjoint.}
\[\text{quantifier rules\footnotemark}\qquad \frac{\Gamma \mid \Delta \mid \Theta
    \vdash t : S}{\Gamma \mid \Delta \mid \Theta \vdash t : \forall \alpha.\;S}
  \qquad \frac{\Gamma \mid \Delta \mid \Theta \vdash t : \forall \alpha.\;
    S}{\Gamma \mid \Delta \mid \Theta \vdash t : S\{\alpha := A\}} \]
\footnotetext{In the introduction rule (left), $\alpha$ must not appear as a
  free variable in $\Gamma$, $\Delta$ and $\Theta$.}
\[\text{functorial promotion rule}\qquad \frac{\varnothing \mid \varnothing \mid
    \Theta \vdash t : \sigma}{\Gamma \mid \oc\Theta, \Delta \mid \Theta' \vdash
    \oc t : \oc\sigma} \]

In these rules, following the conventions established above, $A$ stands for a
linear type, $S$ stands for a strictly linear type and $\sigma$ and $\tau$ stand
for arbitrary types. In particular, in the quantifier elimination rule, $\alpha$
can only be instantiated by a linear type. So, for instance, one cannot give the
type $\oc\beta \multimap \oc\beta$ to $\lambda x.\; x$ through a quantifier
introduction followed by a quantifier elimination; indeed, as one would expect,
the only normal term of this type is $\lambda\oc x.\; \oc{x}$. (Despite this,
the polymorphism is still impredicative.)

\subsection{Encoding regular functions (\cref{thm:ealam-string})}
\label{sec:appendix-ealam-regular}

We fix an input alphabet $\Gamma = \{g_1, \ldots, g_{|\Gamma|}\}$ and an output
alphabet $\Sigma = \{s_1, \ldots, s_{|\Sigma|}\}$.

A first important remark is that the Church encoding in $\EAlam$ consists of an
exponential packaging around an exponential-free term with non-linear free
variables.
\begin{notation}
  We write $t ::_{\Sigma,\alpha} \sigma$, where $\sigma$ is a type which may
  contain the type variable $\alpha$, when the term $t$ uses non-linearly the
  variables $f_i : \alpha \multimap \alpha$ for $i \in \{1,\ldots,|\Sigma|\}$
  and then has type $A$. Formally, using the typing judgment introduced in the
  previous subsection:
  \[t ::_{\Sigma,\alpha} \sigma \iff \varnothing \mid \varnothing \mid f_1 : \alpha
    \multimap \alpha,\, \ldots,\, f_{|\Sigma|} : \alpha \multimap \alpha \vdash
    t : \sigma \]
\end{notation}
Note that $t ::_{\Sigma,\alpha} \sigma \iff \lambda\oc f_1.\; \ldots\;
\lambda\oc f_{|\Sigma|}.\; \oc{t} : (!(\alpha \multimap \alpha))^{|\Sigma|}
\multimap \oc\sigma$.
\begin{definition}
For $w = a_{i_1} \ldots a_{i_n} \in \Sigma^*$, we define $\widetilde{w}$ as follows:
\[ \widetilde{w} = \lambda x.\; f_{i_1}\; (\ldots \;(f_{i_n} \; x)\; \ldots)
  ::_{\Sigma,\alpha} \alpha \multimap \alpha\quad \text{so that}\ \overline{w} =
  \lambda\oc f_1.\; \ldots\; \lambda\oc f_{|\Sigma|}.\; \oc\widetilde{w} :
  \Str_\Sigma \]
\end{definition}

$\widetilde{w}$ is a sort of Church encoding of $w$ \emph{relatively to
  representations of letters provided by the context in the form of non-linear
  variables $f_i : \alpha \multimap \alpha$}. This makes $\alpha \multimap
\alpha$ a kind of \emph{relative type of strings in $\Sigma^*$}, whose advantage
over $\Str_\Sigma$ is that it contains no exponential.

Let us fix a streaming string transducer $(Q,q_I,R,\delta,F)$. We assume without
loss of generality that $Q = \{1, \ldots, |Q|\}$. Recall that $\Fin(|Q|) =
\forall \beta.\; \beta^{|Q|} \multimap \beta$ represents the set of states: the
state $q$ corresponds to the $q$-th projection function $\pi_q = \lambda x_1.\;
\ldots\; \lambda x_{|Q|}.\; x_q$.

According to the discussion in \cref{sec:ealam-string}, $A = \Fin(|Q|) \multimap
(\alpha \multimap \alpha)^{|R|} \multimap (\alpha \multimap \alpha)$ (where
$\alpha$ is a free type variable) should be seen as a type of \emph{linear
  output functions}, that is, of maps $G : Q \to (\Sigma \cup R)^*$ such that
for all $q \in Q$ and $r \in R$, $G(q)$ contains $r$ at most once. (Again, this
is a \emph{relative} type depending on the use of non-linear variables $f_i :
\alpha \multimap \alpha$ given externally.) To formalize this we generalize the
operation $w \rightsquigarrow \widetilde{w}$ to words over $\Sigma \cup R$,
given a fixed enumeration $R = \{r_1, \ldots, r_{|R|}\}$:
\begin{definition}
  For $\omega = c_1 \ldots c_n \in (\Sigma \cup R)^*$, we define
  \[ \widetilde{\omega} = \lambda x.\; \chi(c_1)\; (\ldots \;(\chi(c_n) \; x)\;
    \ldots)\quad\text{where}\ \chi(c) =
    \begin{cases}
      f_i\ \text{for}\ c = a_i \in \Sigma\\
      p_j\ \text{for}\ c = r_j \in R
    \end{cases}\]
  Note that then this is consistent with the previous definition of
  $\widetilde{\omega}$ for $\omega \in \Sigma^*$.

  We also set $\widehat{\omega} = \lambda p_1.\; \ldots \lambda p_{|R|}.\;
  \widetilde{\omega}$.

  Given an output function $G : Q \to (\Sigma \cup R)^*$, we define $\widehat{G}
  = \lambda x.\; x\; \widehat{G(1)} \; \ldots \; \widehat{G(|Q|)}$, so that
  $\widehat{G}$ applied to the $i$-th projection (representing the $i$-th state)
  yields $\widehat{G(i)}$.
\end{definition}
\begin{proposition}
  Because of the linearity constraint for well-typed terms:
  \begin{itemize}
  \item for $\omega \in (\Sigma \cup R)^*$, each register name in $R$ occurs
    \emph{at most once} in $\omega$ if and only if\\
    $p_1 : \alpha \multimap
    \alpha,\, \ldots,\, p_{|R|} : \alpha \multimap \alpha \mid \varnothing \mid
    f_1 : \alpha \multimap \alpha,\, \ldots\,, f_{|\Sigma|} : \alpha \multimap
    \alpha \vdash \widetilde{\omega} : \alpha \multimap \alpha$\\
    or equivalently $\widehat{\omega} ::_{\Sigma,\alpha} (\alpha \multimap \alpha)^{|R|}
      \multimap \alpha \multimap \alpha$
    \item for $G : Q \to (\Sigma \cup R)^*$, $\widehat{G} ::_{\Sigma,\alpha} A$
      iff $G$ is a \emph{linear (i.e.\ copyless)} output function.
  \end{itemize}
\end{proposition}

What we are seeking is an $\EAlam$ term of type $\Str_\Gamma[A] \multimap
\Str_\Sigma[\alpha]$ (with the above type $A$ of linear output functions)
computing the same string function as the SST. We will restrict our search to
terms of the form
\[ \lambda z.\; \lambda\oc f_1\; \ldots\ \lambda\oc f_{|\Sigma|}.\; (\lambda\oc
  h.\; \oc{u})\; (z\; \oc{d_1}\; \ldots\; \oc{d_{|\Gamma|}})\]
where $d_i ::_{\Sigma,\alpha} A \multimap A$ for all $i \in
\{1,\ldots,|\Gamma|\}$, and $h : A \multimap A \mid \varnothing \mid \varnothing
\vdash u : \alpha \multimap \alpha$. It is thanks to the presence of these outer
$\lambda\oc f_i$ binding non-linearly the $f_i : \alpha \multimap \alpha$ that
the various relative representations we are manipulate are meaningful. The idea
is that, since
\[ (\lambda\oc h.\; \oc{u})\; (\overline{w}\; \oc{d_1}\; \ldots\;
  \oc{d_{|\Gamma|}}) =_\beta \oc{}u\{h:=(\lambda x.\; d_{i_1}\; (\ldots(d_{i_n}\;
  x)\ldots))\}\quad \text{for $w = g_{i_1} \ldots g_{i_n} \in \Gamma^*$,}
\]
we can take:
\begin{itemize}
\item $d_i$ to be the representation of the action of the transition for the
  letter $g_i \in \Gamma$ on the output function, that is, what we called
  $\delta^O(g_i, -)$ (\cref{def:delta-o});
\item $u$ to be a term applying $h$ to a representation of the SST's original
  output function $F$, and then using the result to extract the image of the
  input word, following the recipe of \cref{prop:delta-o}.
\end{itemize}

Formally, what we want for $d_i$ ($i \in \{1,\ldots,|\Gamma|\}$) is $d_i
::_{\Sigma,\alpha} A \multimap A$ and $d_i\;\widehat{G} =_\beta
\widehat{\delta^O(g_i,G)}$. The typing condition tells us to look for a term of
the form $d_i = \lambda z.\; \lambda y.\; y\; t_{i,1}\; \ldots\; t_{i,q}$, where
for all $q \in Q$, we have $z : A,\, y : \Fin(|Q|) \mid \varnothing \mid f_1 :
\alpha \multimap \alpha,\,\ldots,\, f_{|\Sigma|} : \alpha \multimap \alpha
\vdash t_{i,q} : A$.

Let $q \in Q$. We also want $t_{i,q}\{z := \widehat{G}\}$ to somehow represent
$\delta^O(a,G)(q)$ for $a = g_i \in \Gamma$; by definition, this equals
$s_{a,q}^*(G(q'_{a,q}))$, where $(q'_{a,q},s_{a,q}) = \delta(q,a)$ (cf.\
\cref{def:delta-o}). We use a technique analogous to the proof of
\cref{lem:stlam-morphism} to express the application of a monoid morphism on a
Church-encoded string (here, the concerned string is $z\;\pi_{q'_{q,a}}$):
\[ t_{i,q} = \lambda p_1.\; \ldots\; \lambda p_{|R|}.\;
  z\;\pi_{q'_{q,a}}\;\widetilde{s_{a,q}(r_1)}\;\ldots\;\widetilde{s_{a,q}(r_{|R|})} \]
In order for this to be well-typed, various linearity conditions must be
satisfied. In particular, each $p_j$ ($j \in \{1,\ldots,|R|\}$) must occur at
most once in all $\widetilde{s_{a,q}(r_1)},\ldots,\widetilde{s_{a,q}(r_{|R|})}$.
By definition of the encoding $\widetilde{(\;\cdot\;)}$, this is the case iff
each $r_j$ occurs at most once in all $s_{a,q}(r)$ for $r \in R$. This none
other than the copyless assignment condition for transitions.

This concludes the definition of $d_i$. As for $u$, it is set to $u = h\;
\widehat{F}\; \pi_{q_I}\; (\lambda x.\; x)\; \ldots\; (\lambda x.\; x)$, with
$|R|$ times $(\lambda x.\; x)$. Using \cref{prop:delta-o}, one can check that
the term we get in the end -- that is, $\lambda z.\; \lambda\oc f_1\; \ldots\
\lambda\oc f_{|\Sigma|}.\; (\lambda\oc h.\; \oc{u})\; (z\; \oc{d_1}\; \ldots\;
\oc{d_{|\Gamma|}})$ -- computes the right function.

\subsection{Encoding composition by substitution
  (\cref{prop:ealam-cbs})}
\label{sec:appendix-ealam-polyseq}

We must show that if $f : \Gamma^* \to I^*$ and $g_i : \Gamma^* \to \Sigma^*$
are defined by some respective $\EAlam$ terms $t : \oc\Str_\Gamma \multimap
\oc\Str_I$ and $u_i : \oc\Str_\Gamma \multimap \oc\Str_\Sigma$ (for $i \in I$,
assuming w.l.o.g.\ that $I = \{1,\ldots,|I|\}$), then their composition by
substitution $\CbS(f,(g_i)_{i \in I})$ is also definable as a term of type
$\oc\Str_\Gamma \multimap \oc\Str_\Sigma$. The term we use for that purpose is
\[ \lambda\oc s.\; (\lambda\oc x.\; \lambda\oc y_1.\; \ldots \lambda\oc
  y_{|I|}.\; \oc{s'})\;
  (t\;\oc{s})\;(u_1\;\oc{s})\;\ldots\;(u_{|I|}\;\oc{s})\; \]
\[ \text{where}\ s' = \lambda\oc f_1.\; \ldots\; \lambda\oc f_{|\Sigma|}.\; x\;
  (y_1\; \oc{f_1}\; \ldots\; \oc{f_{|\Sigma|}})\; \ldots\; (y_{|I|}\; \oc{f_1}\;
  \ldots\; \oc{f_{|\Sigma|}}) \]

\subsection{Encoding regular tree functions (\cref{thm:ealam-tree})}
\label{sec:appendix-ealam-tree}

Regular tree functions are the functions computed by bottom-up ranked tree
transducers, whose definition we now give in its entirety.

\begin{definition}[{\cite{BRTT}}]
  A \emph{conflict relation} is a binary reflexive and symmetric relation.

  Let $\incoh$ be a conflict relation over $V \cup V'$, and $E \in
  \ExprBT(\Sigma,V,V')$. The expression $E$ is \emph{consistent with~$\incoh$}
  when
  \begin{itemize}
  \item each variable in $V \cup V'$ appears at most once in $E$;
  \item for all $x,y \in V \cup V'$, if $x \neq y$ and $x \incoh y$, then $E$
    does not contain both $x$ and $y$.
  \end{itemize}
  Consistency with $\incoh$ is defined in the same way for expressions in
  $\ExprdBT(\Sigma,V,V')$.

  A \emph{bottom-up ranked tree transducer (BRTT)} is a register tree transducer
  $(Q,q_I,R,R',F,\delta)$ endowed with a conflict relation $\incoh$ on $R \cup
  R'$, such that:
  \begin{itemize}
  \item for all $q \in Q$, the expression $F(q)$ is consistent with $\incoh$;
  \item for all $\varepsilon : R \to \ExprBT(\Sigma,R_{\ttl\ttr},R'_{\ttl\ttr})$
    and $\varepsilon' : R' \to \ExprdBT(\Sigma,R_{\ttl\ttr},R'_{\ttl\ttr})$, if
    there exist $q, q_\ttl, q_\ttr \in Q$ and $a \in \Gamma$ such that $(q,
    \varepsilon, \varepsilon') = \delta(q_\ttl,q_\ttr,a)$, then
    \begin{itemize}
    \item all $\varepsilon(r)$ for $r \in R$ and all $\varepsilon'(r')$ for $r'
      \in R'$ are consistent with $\incoh$;
    \item if $x_1, x_2, y_1, y_2 \in R \cup R'$, $x_1 \incoh x_2$ and, for some
      $z \in \{\ttl,\ttr\}$, $(x_1, z)$ appears in\footnote{By $\varepsilon \cup
        \varepsilon'$ we mean the map $R \cup R' \to
        \ExprBT(\Sigma,R_{\ttl\ttr},R'_{\ttl\ttr}) \cup
        \ExprdBT(\Sigma,R_{\ttl\ttr},R'_{\ttl\ttr})$ induced in the obvious way
        by $\varepsilon$ and $\varepsilon'$ -- recall that $R \cup R'$ is a
        disjoint union.} $(\varepsilon \cup \varepsilon')(y_1)$ and $(x_2, z)$
      appears in $(\varepsilon \cup \varepsilon')(y_2)$, then $y_1 \incoh y_2$.
    \end{itemize}
  \end{itemize}
\end{definition}

This is indeed a kind of generalized linearity condition: when $\incoh \;=
\{(x,x) \mid x \in R \cup R'\}$, we recover the notion of copyless assignment
used for streaming string transducers. Before we start translating BRTTs into
$\EAlam$ terms, we first reformulate this relaxed linearity using
non-conflicting subsets.
\begin{definition}
  Let $X$ be a set endowed with a conflict relation $\incoh$. A subset $P
  \subseteq X$ is said to be \emph{non-conflicting} if $\forall x,y \in P,\, x =
  y \lor x \not\incoh y$. We write $P \sqsubseteq X$.
\end{definition}
\begin{remark}
  As the reader might have noticed, the notations are meant to draw parallels to
  the structure of coherence spaces (a simple semantics of linear logic).
\end{remark}
\begin{notation}
  For $E \in \ExprBT(\Sigma,V,V') \cup \ExprdBT(\Sigma,V,V')$, we write
  $\mathcal{V}(E)$ for the set of variables occurring in $E$, so that
  $\mathcal{V}(E) \subseteq V \cup V'$.
\end{notation}
\begin{proposition}
  An expression $E \in \ExprBT(\Sigma,V,V') \cup \ExprdBT(\Sigma,V,V')$ is
  consistent with a conflict relation over $V \cup V'$ if and only if it is
  \emph{linear} (each variable appears at most once) and $\mathcal{V}(E)
  \sqsubseteq V \cup V'$ (that is, $\mathcal{V}(E)$ is \emph{non-conflicting}).

  A register tree transducer $(Q,q_I,R,R',F,\delta)$ endowed with a conflict
  relation $\incoh$ on $R \cup R'$ is a BRTT if and only if (recall that
  $R_{\ttl\ttr} = R \times \{\ttl,\ttr\}$):
  \begin{itemize}
  \item for all $q \in Q$, $F(q)$ is linear and $\mathcal{V}(F(q))$ is
    non-conflicting;
  \item for all $\varepsilon : R \to \ExprBT(\Sigma,R_{\ttl\ttr},R'_{\ttl\ttr})$
    and $\varepsilon' : R' \to \ExprdBT(\Sigma,R_{\ttl\ttr},R'_{\ttl\ttr})$, if
    there exist $q, q_\ttl, q_\ttr \in Q$ and $a \in \Gamma$ such that $(q,
    \varepsilon, \varepsilon') = \delta(q_\ttl,q_\ttr,a)$, then
    \begin{itemize}
    \item all $\varepsilon(r)$ for $r \in R$ and all $\varepsilon'(r')$ for $r'
      \in R'$ are linear;
    \item for all non-conflicting $P \sqsubseteq R \cup R'$, the sets
      $\mathcal{V}((\varepsilon \cup \varepsilon')(y))$ for $y \in P$ are
      pairwise disjoint, and their union $\bigcup_{y\in P}
      \mathcal{V}((\varepsilon \cup \varepsilon')(y))$ is non-conflicting in
      $R_{\ttl\ttr} \cup R'_{\ttl\ttr}$ -- where the conflict relation of the
      latter is defined so that\footnote{Pursuing the analogy with coherence
        spaces, we have, morally, $R_{\ttl\ttr} \cong (R \otimes \{\ttl\}) \with
        (R \otimes \{\ttr\})$.} there is never a conflict between $(x_1,\ttl)$
      and $(x_2,\ttr)$ for $x_1,x_2 \in R \cup R'$.
    \end{itemize}
  \end{itemize}
\end{proposition}

The moral of the story until now is that, while it is not true that the
transition function performs copyless assignments, one can say instead that:
\begin{itemize}
\item for every non-conflicting set of register names $P \sqsubseteq R \cup R'$,
  the contents of the registers in $P$ after a transition are obtained linearly
  (by copyless assignment) from the contents of a non-conflicting subset of
  $R_{\ttl\ttr} \cup R'_{\ttl\ttr}$;
\item in the end, depending on the final state, one such $P \sqsubseteq R \cup
  R'$ is used linearly to produce the output.
\end{itemize}
We must now show that $\EAlam$ is expressive enough to accomodate this
variation on linearity.

Let $(Q,q_I,R,R',F,\delta, \incoh)$ be a BRTT. Analogously to the encoding of
SSTs in $\EAlam$, we encode this as a term of the form $\lambda z.\; \lambda\oc
f_1.\; \ldots\; \lambda\oc f_{|\Gamma|}.\; \lambda\oc{}x.\; (\ldots)$. Thus, we
will be able to manipulate representations of data relatively to non-linear
variables $f_i : \alpha \multimap \alpha \multimap \alpha$ (representing
$g_i\langle -,- \rangle$ for $g_i \in \Gamma$) and $x : \alpha$ (representing
$\langle \rangle$): abbreviating $\vec{f} : \alpha \multimap \alpha \multimap
\alpha$ for $f_1 : \alpha \multimap \alpha \multimap \alpha,\, \ldots,\,
f_{|\Sigma|} : \alpha \multimap \alpha \multimap \alpha$, there are natural
encodings
\begin{align*}
T \in \BinTree(\Sigma) &\rightsquigarrow \varnothing \mid \varnothing \mid
  \vec{f} : \alpha \multimap \alpha \multimap \alpha,\, x : \alpha \vdash
  \widetilde{T} : \alpha\\
T' \in \dBinTree(\Sigma) &\rightsquigarrow \varnothing \mid \varnothing
  \mid \vec{f} : \alpha \multimap \alpha \multimap \alpha,\, x : \alpha \vdash
  \widetilde{T'} : \alpha \multimap \alpha\\
E \in \ExprBT(\Sigma,V,V') &\rightsquigarrow \varnothing \mid \varnothing
  \mid \vec{f} : \alpha \multimap \alpha \multimap \alpha,\, x : \alpha \vdash
  \widetilde{E} : \alpha^{|V|} \multimap (\alpha \multimap \alpha)^{|V'|}
  \multimap \alpha\\
E' \in \ExprdBT(\Sigma,V,V') &\rightsquigarrow \varnothing \mid \varnothing
  \mid \ldots \vdash
  \widetilde{E'} : \alpha^{|V|} \multimap (\alpha \multimap \alpha)^{|V'|}
  \multimap (\alpha \multimap \alpha)
\end{align*}
For a BRTT with $\incoh \;= \{(x,x) \mid x \in R \cup R'\}$), i.e.\ with truly
copyless assignments, the relative type of configurations would be
\[ A = \Fin(|Q|) \otimes \alpha^{\otimes|R|} \otimes (\alpha \multimap
  \alpha)^{\otimes|R'|} \qquad\text{where}\ B^{\otimes m} = B
  \otimes \ldots \otimes B \]
The transition after reading some label $a \in \Sigma$ in a node must be of the
type $A \multimap A \multimap A$. Morally, this is isomorphic to $(A \otimes A)
\multimap A$, and since $|R_{\ttl\ttr}| = 2|R|$,
\[ A \otimes A \cong \Fin(|Q|) \otimes \Fin(|Q|) \otimes
  \alpha^{\otimes|R_{\ttl\ttr}|} \otimes (\alpha \multimap
  \alpha)^{|R'_{\ttl\ttr}|} \]
These isomorphisms of linear logic are not quite reflected as actual type
isomorphisms in $\EAlam$, since the multiplicative conjunction $\otimes$ does
not exist as a primitive, and we use instead a second-order encoding already
exploited in~\cite{Benedetti,ealreg}. But they illustrate the reason why
$\delta(-,-,a)$ ($a \in \Gamma$) can be turned into a term of type $A \multimap
A \multimap A$; in particular the type
\[ \alpha^{\otimes|R_{\ttl\ttr}|} \otimes (\alpha \multimap
  \alpha)^{\otimes|R'_{\ttl\ttr}|} \multimap \alpha^{\otimes|R|} \otimes (\alpha
  \multimap \alpha)^{\otimes|R'|} \]
corresponds to the $(\varepsilon \cup \varepsilon') : R \cup R' \to
\ExprBT(\Sigma,R_{\ttl\ttr},R'_{\ttl\ttr}) \cup
\ExprdBT(\Sigma,R_{\ttl\ttr},R'_{\ttl\ttr})$ that was mentioned in the
definition of BRTTs. And the function arrow can be linear because the
register update $(\varepsilon \cup \varepsilon')$ is copyless.

We now come to the case of a BRTT with an arbitrary conflict relation. The
relaxed linearity of $(\varepsilon \cup \varepsilon')$ is then manifested as the
fact that for all non-conflicting $P \sqsubseteq R \cup R'$, one can represent
(relatively to $f_i$ and $x$) its action to produce the new contents of $P$ as
an $\EAlam$ term of type
\[ \alpha^{\otimes|S \cap R_{\ttl\ttr}|} \otimes (\alpha \multimap
  \alpha)^{\otimes|S \cap R'_{\ttl\ttr}|} \multimap \alpha^{\otimes|P \cap R|}
  \otimes (\alpha \multimap \alpha)^{\otimes|P \cap R'|}\quad S =
  \bigcup_{y\in P} \mathcal{V}((\varepsilon \cup \varepsilon')(y)) \]
It is important to observe that $S$ is also non-conflicting, thanks to a
previous proposition. The entirety of $(\varepsilon \cup \varepsilon')$ can
therefore be faithfully represented by an $\EAlam$ term of type
\[ \bigwith_{S \sqsubseteq R_{\ttl\ttr} \cup R'_{\ttl\ttr}} \left(
    \alpha^{\otimes|S \cap R_{\ttl\ttr}|} \otimes (\alpha \multimap
    \alpha)^{\otimes|S \cap R'_{\ttl\ttr}|} \right) \multimap \bigwith_{P
    \sqsubseteq R \cup R'} \left( \alpha^{\otimes|P \cap R|} \otimes (\alpha
    \multimap \alpha)^{\otimes|P \cap R'|} \right) \]
using the encoding of the additive conjunction in $\EAlam$
\[ A_1 \with \ldots \with A_m := \forall \gamma.\; (\forall \beta.\; \beta
  \multimap (\beta \multimap A_1) \multimap \ldots \multimap (\beta \multimap
  A_m) \multimap \gamma)\multimap \gamma \]
This explains the use of the type of configurations
\[ A = \Fin(|Q|) \otimes \bigwith_{P \sqsubseteq R \cup R'}
  \left( \alpha^{\otimes|P \cap R|} \otimes
      (\alpha \multimap \alpha)^{\otimes|P \cap R'|} \right) \]
for general BRTTs. (To recover the left side of the previous type from $A
\otimes A$, use the canonical function
$\displaystyle (A_1 \with \ldots \with A_m) \otimes (B_1 \with \ldots \with
B_m) \multimap \bigwith_{1 \leq i,j \leq m} (A_i \otimes B_j)$.)

At the end, one must extract the output from the final configuration.
Fortunately, for any state, the corresponding output expression in
$\ExprBT(\Sigma,R,R')$ only involves a non-conflicting set $P \sqsubset R \cup
R'$ of variables. Thus, one can project the final configuration to retrieve a
datum of type $\alpha^{\otimes|P \cap R|} \otimes (\alpha \multimap
\alpha)^{\otimes|P \cap R'|}$ for this specific $P$; this is sufficient to
determine the output tree (represented by an element of type $\alpha$).

\end{document}